\newif\ifIEEE
\newif\ifISIT
\IEEEfalse     
\IEEEtrue
\ISITfalse

\ifISIT
    \IEEEtrue
\fi
\ifIEEE
    \ifISIT
        \documentclass[conference,final,letterpaper]{IEEEtran}
    \else
        \documentclass[journal,final,letterpaper]{IEEEtran}
    \fi
    \newcommand{\bibauthor}[1]{#1}
    \newcommand{\bibpaper}[1]{``#1''}
    \newcommand{\Footnotetext}[2]
    {
        \begin{figure}[!b]
        \footnotesize\vspace{-3ex}\hrulefill\hfill
        \makebox[0em]{}\hfill\makebox[0em]{}%
                                          \par${}^{#1}$ #2\vspace{-.6ex}
        \end{figure}
        \addtocounter{figure}{0}
     }
\else
    \documentclass[12pt]{article}
    \newcommand{\bibauthor}[1]{\textsc{#1}}
    \newcommand{\bibpaper}[1]{\textsl{#1}}
    \newenvironment{IEEEkeywords}{\begin{small}%
                                  \textbf{Index Terms} ---}{\end{small}}
\fi
\usepackage{ifthen}
\usepackage{amsfonts,bm,bbm}
\usepackage{amsthm,amsmath,amssymb}
\usepackage{pict2e}
\usepackage{color}
\usepackage{float}
\newcommand{\bibbook}[1]{\textit{#1}}

\newcommand{\bibperiodical}[1]{\textit{#1}}

\ifISIT
\fi
\newtheorem{theorem}{\indent Theorem}

\newtheorem{lemma}[theorem]{\indent Lemma}
\newtheorem{corollary}[theorem]{\indent Corollary}

\theoremstyle{remark}
\newtheorem{remark}{\indent Remark}
\theoremstyle{definition}
\newtheorem{example}{\indent Example}
\renewcommand{\mathbf}[1]{{\bm{#1}}}     
\newenvironment{myalgorithm}{%
        \begin{minipage}{\columnwidth}
        \makebox[0ex]{}\hrulefill\makebox[0ex]{}\\*%
             }{%
             \vspace{-1ex}%
             \makebox[0ex]{}\hrulefill\makebox[0ex]{}\end{minipage}}
\newcommand{\bldzero}{{\mathbf{0}}}
\newcommand{\bldone}{{\mathbf{1}}}
\newcommand{\blda}{{\mathbf{a}}}
\newcommand{\bldb}{{\mathbf{b}}}
\newcommand{\bldc}{{\mathbf{c}}}
\newcommand{\blde}{{\mathbf{e}}}
\newcommand{\bldg}{{\mathbf{g}}}
\newcommand{\bldh}{{\mathbf{h}}}
\newcommand{\bldu}{{\mathbf{u}}}
\newcommand{\bldv}{{\mathbf{v}}}
\newcommand{\bldx}{{\mathbf{x}}}
\newcommand{\bldy}{{\mathbf{y}}}
\newcommand{\bldz}{{\mathbf{z}}}
\newcommand{\Int}[1]{{\left[{#1}\right\rangle}}
\newcommand{\p}[1]{{\mathrm{#1}}}
\newcommand{\Field}{{\mathbb{F}}}
\newcommand{\Integers}{{\mathbb{Z}}}
\newcommand{\Realfield}{{\mathbb{R}}}
\newcommand{\Sphere}{{\mathbb{S}}}
\newcommand{\Height}{{\mathsf{h}}}
\newcommand{\rank}{{\mathrm{rank}}}
\newcommand{\diag}{{\mathrm{diag}}}
\newcommand{\code}{{\mathcal{C}}}
\renewcommand{\Subset}{{\mathcal{S}}}
\newcommand{\Support}{{\mathsf{Supp}}}
\newcommand{\sgn}{{\mathrm{sgn}}}
\newcommand{\weight}{{\mathsf{w}}}
\renewcommand{\AA}{{\mathcal{A}}}
\newcommand{\EE}{{\mathcal{E}}}
\newcommand{\GG}{{\mathcal{G}}}
\newcommand{\II}{{\mathcal{I}}}
\newcommand{\JJ}{{\mathcal{J}}}
\newcommand{\PP}{{\mathcal{P}}}
\newcommand{\TT}{{\mathcal{T}}}
\newcommand{\XX}{{\mathcal{X}}}
\newcommand{\ico}{{\text{\normalfont ico}}}
\newcommand{\dod}{{\text{\normalfont dod}}}
\newcommand{\rut}{{\text{\normalfont Rut}}}

\newcommand{\Title}{On the Height Profile of Analog
            Error-Correcting Codes}
\newcommand{\Namea}{Ron M. Roth}
\newcommand{\Nameb}{Ziyuan Zhu}
\newcommand{\Namec}{Changcheng Yuan}
\newcommand{\Named}{Paul H. Siegel}
\newcommand{\Namee}{Anxiao Jiang}
\newcommand{\Addressa}{%
             Computer Science Department, Technion, Haifa, Israel}
\newcommand{\Addressb}{%
             Department of Electrical and Computer Engineering,
             University of California San Diego, La Jolla, CA, USA}
\newcommand{\Addressc}{%
             Department of Computer Science and Engineering,
             Texas A\&M University, College Station, TX, USA}
\newcommand{\Addressd}{\Addressb}
\newcommand{\Addresse}{\Addressc}
\newcommand{\Emaila}{ronny@cs.technion.ac.il}
\newcommand{\Emailb}{ziz050@ucsd.edu}
\newcommand{\Emailc}{ericycc@tamu.edu}
\newcommand{\Emaild}{psiegel@ucsd.edu}
\newcommand{\Emaile}{ajiang@tamu.edu}
\newcommand{\Commenta}{The work of R.M. Roth was done in part while
                    visiting the Center for Memory and Recording
                    Research, UC San Diego}
\newcommand{\Grant}{This work was supported in part by
                    NSF Grants Nos.~CCF-2212437,
                    CCF-2416361, and CCF-2416362,
                    by Grant No.~2023786 from
                    the United-States--Israel Binational
                    Science Foundation (BSF),
                    and by Grant No.~2369/24 from
                    the Israel Science Foundation}
\newcommand{\Thnxa}{\Namea\ is with the \Addressa.
                    \Commenta\ (email: \Emaila).}
\newcommand{\Thnxb}{\Nameb\ is with the \Addressb\ (email: \Emailb).}
\newcommand{\Thnxc}{\Namec\ is with the \Addressc\ (email: \Emailc).}
\newcommand{\Thnxd}{\Named\ is with the \Addressd\ (email: \Emaild).}
\newcommand{\Thnxe}{\Namee\ is with the \Addresse\ (email: \Emaile).}

\begin{document}
\ifIEEE
    \title{\Title}
        \ifISIT
\else
            \author{\Namea \quad\quad
                    \Nameb \quad\quad
                    \Namec \quad\quad
                    \Named \quad\quad
                    \Namee
                    \thanks{\Grant.}
                    \thanks{\Thnxa}
                    \thanks{\Thnxb}
                    \thanks{\Thnxc}
                    \thanks{\Thnxd}
                    \thanks{\Thnxe}
            }
        \fi
\else
    \title{\textbf{\Title}\thanks{\Grant.}}
    \author{\textsc{\Namea}\thanks{\Thnxa}
    \and
            \textsc{\Nameb}\thanks{\Thnxb}
    \and
            \textsc{\Namec}\thanks{\Thnxc}
    \and
            \textsc{\Named}\thanks{\Thnxd}
    \and
            \textsc{\Namee}\thanks{\Thnxe}
    }
\fi

\maketitle

\begin{abstract}
In recent work, it has been shown that
maintaining reliability in analog vector--matrix multipliers
can be modeled as the following coding problem.
Vectors in $\Realfield^k$ are encoded into codewords of
a linear $[n,k,d]$ code $\code$ over $\Realfield$.
For prescribed positive reals $\delta < \Delta$,
additive errors of magnitude${} \le \delta$
are tolerable and need no handling,
yet outlying errors of magnitude${} > \Delta$
are to be located or detected.
The trade-off between the ratio $\Delta/\delta$
and the number of outlying errors that can be handled is determined
by the \emph{height profile} of $\code$; as such,
the height profile provides a finer description
of the error handling capability of $\code$,
compared to the minimum distance $d$, which only determines
the number of correctable errors.
This work contains a further study of the notion of the height profile.
Several characterizations of the height profile are presented,
thereby yielding methods for computing it.
\ifISIT
\else
The starting point is formulating this computation
as an optimization problem that is solved by a set of linear programs.
This, in turn, leads to a combinatorial characterization of
the height profile as a maximum (or max--min) over
a certain finite set of codewords of $\code$.
Moreover, this characterization is shown to have
a simple geometric interpretation when the columns of
the generator matrix of $\code$ all have the same $L_2$ norm.
Through examples of several code families, it is demonstrated
how the results herein can be used to compute the height profile
explicitly.
\fi
\end{abstract}

\ifISIT
\else
\begin{IEEEkeywords}
Fault-tolerant computing,
least absolute deviations,
linear codes over the real field,
linear programming,
Platonic solids,
spherical codes,
unit-norm tight frame,
vector--matrix multiplication.
\end{IEEEkeywords}
\fi

\section{Introduction}
\label{sec:introduction}

Vector--matrix multiplication (VMM) is a very common computational task
found in many applications, and there are quite a few
hardware implementations of analog VMM accelerators, particularly
designs that are based on memristive crossbars~%
\cite{AHCNFWFHSRM},%
\cite{HSKGDGLGWY},%
\cite{HGLLGMDJWYXS},%
\cite{SNMBSHWS},%
\cite{TANSB}.
Such devices, however, are susceptible to various types of errors.
As shown in recent work~\cite{Roth1},\cite{Roth2},
maintaining reliability in such devices
can be modeled as the following coding problem.
Blocks in $\Realfield^k$ are mapped one-to-one to codewords
of a linear $[n,k,d]$ code $\code$ over $\Realfield$.
For prescribed positive reals $\delta < \Delta$,
additive errors whose values are within the interval
$[-\delta,\delta]$ are tolerable and need no handling;
on the other hand, outlying errors, whose values are
outside the interval $[-\Delta,\Delta]$, are to be located or detected.
The correction capability of the code $\code$ is determined by
the following parameters:
(i)~the number $\tau$ of outlying errors that can be located,
(ii)~the number $\tau + \sigma$ of outlying errors that can be detected,
and (iii)~the ratio $\Delta/\delta$ (which should be as small
as possible). In the classical coding setting,
the inequality $2 \tau + \sigma < d$ is necessary and sufficient
in order to correct any $\tau$ errors and detect
any $\tau + \sigma$ errors. Herein, a finer parametrization of $\code$
is required that also determines the smallest attainable ratio
$\Delta/\delta$; this is done by means of
the \emph{height profile} of the code, which is a list
of positive (possibly infinite) real values $(\Height_m(\code))_m$,
whose precise definition will be recalled below.
In terms of the height profile, we then get
the following necessary and sufficient relationship
between the parameters $\tau$, $\sigma$, and $\Delta/\delta$:
\[
\Delta/\delta \ge 2 (\Height_{2\tau + \sigma}(\code) + 1) .
\]
The aim of this work is providing several characterizations
of the height profile of a linear code $\code$.

We next introduce some notation and definitions.
For integers $\ell \le n$, we denote by $\Int{\ell:n}$
the integer subset
$\left\{ z \in \Integers \,:\, \ell \le z < n \right\}$.
When $\ell = 0$ we use the shorthand notation $\Int{n}$,
and we will typically use
$\Int{n}$ to index the entries of vectors in $\Realfield^n$.
Similarly, the entries of an $r \times n$ matrix $H = (H_{i,j})$ will be
indexed by $(i,j) \in \Int{r} \times \Int{n}$.
The $i$th row of $H$ will be denoted by $[H]_i$,
and for a subset $\II \subseteq \Int{n}$, the notation $(H)_\II$
will stand for the $r \times |\II|$ submatrix of $H$
that is formed by the columns that are indexed by $\II$
(this notation extends to vectors as well, seen as single-row matrices).
We denote the support of a real vector $\bldx$ by $\Support(\bldx)$
and its $L_p$ norm by $\| \bldx \|_p$,
with $\| \bldx \|_\infty$ standing for the largest magnitude
(i.e., absolute value) of any entry of $\bldx$.
The unit (Euclidean) hypersphere in $\Realfield^k$, namely, the set
$\left\{ \bldx \in \Realfield^k \,:\, \| \bldx \|_2 = 1 \right\}$,
will be denoted by $\Sphere^{k-1}$.

Given a nonzero vector $\bldx = (x_j)_{j \in \Int{n}} \in \Realfield^n$,
a \emph{sorting permutation} of $\bldx$ is a permutation
$\pi : \Int{n} \rightarrow \Int{n}$ that
\ifISIT
\else
sorts\footnote{%
A sorting permutation is not unique if $\bldx$ has multiple entries
with the same magnitude, yet this will not pose
any impediment in the sequel.}
\fi
the entries of $\bldx$ according to descending magnitudes:
\[
|x_{\pi(0)}| \ge
|x_{\pi(1)}| \ge \ldots \ge |x_{\pi(n-1)}| .
\]
Given an integer $m \in \Int{n}$,
the \emph{$m$-height} of $\bldx$,
denoted $\Height_m(\bldx)$, is
\ifISIT
\else
defined by
\[
\Height_m(\bldx) = \left| \frac{x_{\pi(0)}}{x_{\pi(m)}} \right|
\]
(where $\Height_m(\bldx) = \infty$ when $x_{\pi(m)} = 0$).
We also define the set\footnote{%
Note that $\PP_m(\bldx)$ is uniquely defined, even when $\bldx$ has
several sorting permutations.}
\fi
\[
\PP_m(\bldx)
= \left\{ j \in \Int{n} \,:\, |x_j| = |x_{\pi(m)}| \right\} .
\]

The $m$-height of a linear $[n,k,d]$
\ifISIT
\else
code\footnote{%
Hereafter, we always assume that the code is nontrivial, i.e., $k > 0$.}
\fi
$\code$ over $\Realfield$
is defined by
\[
\Height_m(\code) = \max_{\bldc \in \code \setminus \{ \bldzero \}}
\, \Height_m(\bldc) ,
\]
and a nonzero codeword of $\code$ is called \emph{$m$-extremal}
if its $m$-height equals $\Height_m(\code)$.
The list $(\Height_m(\code))_{m \in \Int{n}}$ is called
the \emph{height profile} of $\code$.
The minimum distance $d$ of $\code$ is related to its height profile by
\begin{equation}
\label{eq:height-distance}
d = \min \Bigl\{ m \in \Int{n{+}1} \,:\, \Height_m(\code)=\infty \Bigr\}
\end{equation}
(where we formally define $\Height_n(\code) = \infty$).
\ifISIT
\else
The minimum-weight codewords of $\code$ are its codewords $\bldc$
whose Hamming weight, $\weight(\bldc)$, equals $d$.
\fi

For a subset $\II \subseteq \Int{n}$ we let
$\overline{\II}$ be the complement set
$\Int{n} \setminus \II$. The ambient set $\Int{n}$ will always be
the index set of the coordinates of a linear code
and will always be understood from the context.

Given a linear $[n,k]$ code $\code$ over $\Realfield$
and a subset $\II \subseteq \Int{n}$,
we let $(\code)_\II$ denote
\ifISIT
\else
the linear code that is obtained
by puncturing $\code$ on $\overline{\II}$:
\fi
\[
(\code)_\II =\Bigl\{ (\bldc)_\II \,:\, \bldc \in \code \Bigr\} .
\]
We also define $\code\vert_\II$ to be the set of all codewords
$\bldc \in \code$ whose projections $(\bldc)_\II$
are in the hypercube $[-1,1]^{|\II|}$, i.e.,
\begin{equation}
\ifISIT
\else
\label{eq:codesubset}
\fi
\code\vert_\II = \Bigl\{
\bldc = (c_j)_{j \in \Int{n}} \in \code \,:\,
\textrm{$|c_j| \le 1$ for all $j \in \II$}
\Bigr\} .
\end{equation}

As said earlier, our goal in this work is to find several
characterizations of the height profile of a linear code,
thereby presenting methods for calculating it.
Now, the relationship~(\ref{eq:height-distance}) implies that computing
the height profile of a generic linear code is at least as hard as
computing its minimum distance; still, since the codes of interest
are typically very structured, our characterizations
provide tools for analyzing the height profile of such codes
(as we demonstrate through several examples).

\ifISIT
\else
Our starting point will be casting the problem of computing
the $m$-height of a linear code as a maximization over the outputs of
a set of linear programs (Section~\ref{sec:LP}).
This approach was recently used in~\cite{Jiang} when searching
for codes with favorable height profiles.\footnote{%
Tables such as~\cite[Table~III]{Jiang} and~\cite{YJRZS}
can be used as a yardstick
to determine how good the height profile is, as we currently do not know
of negative results (i.e., general lower bounds) on the entries of
the height profile.}
The algorithm that we present here
is simpler; moreover, through the dual programs, we obtain
a necessary condition on the structure of $m$-extremal codewords
(this condition will turn out to be useful in subsequent sections).

In Section~\ref{sec:combinatorial},
we present several characterizations of
the $m$-height of a linear code as
a maximum (or max--min) over a certain \emph{finite} set of codewords.
These characterizations, in turn, allow us to obtain exact expressions
for the height profile of several (structured) family of codes.

In the context of searching for codes with
favorable height profiles, almost all constructions so far
in the literature are codes with parity-check matrices whose columns are
all unit vectors. In Section~\ref{sec:spherical}, we focus on such
codes and their duals. In particular, for the dual codes,
we show that our condition on the structure of $m$-extremal codewords
has a particularly simple geometric interpretation.
This, in turn, allows us to completely characterize
the $m$-extremal codewords in several cases of interest.
\fi

\section{Height via linear programming}
\label{sec:LP}

The following result is the first among several characterizations
that we provide for the $m$-height of a linear code.
It will serve as a basis for characterizing
the $m$-height of a linear code as the largest among
the solutions of a certain set of linear programs.
The approach of computing the $m$-height of codes
through linear programming was first introduced
\ifISIT
\else
in~\cite{Jiang} and the
\fi
algorithm to be presented in Figure~\ref{fig:LP} below
can be seen as a simplified version of the algorithm in~\cite{Jiang}.
Our new algorithm was used to obtain
the improved table in~\cite{YJRZS}.

\begin{theorem}
\label{thm:LP}
Let $\code$ be a linear $[n,k,d]$ code over $\Realfield$
and let $m \in \Int{d}$. Then
\begin{equation}
\label{eq:LP1}
\Height_m(\code) =
\max_\Subset
\max_{\bldx \in \code \vert_{\overline{\Subset}}} \| \bldx \|_\infty ,
\end{equation}
where $\Subset$ ranges over all $m$-subsets of $\Int{n}$.
\end{theorem}

\begin{proof}
We first show that $\Height_m(\code)$ is bounded from above
by the right-hand side of~(\ref{eq:LP1}).
Let $\bldc = (c_j)_{j \in \Int{n}}$ be an $m$-extremal codeword
in $\code$ and let $\pi(\cdot)$ be a sorting permutation of $\bldc$;
without loss of generality we can assume that
$|c_{\pi(m)}| = 1$, in which case
\ifISIT
\else
\[
\Height_m(\bldc) = \| \bldc \|_\infty .
\]
\fi
Since there are no more than $m$ entries in $\bldc$ of magnitude
greater than $1$, there exists at least
one $m$-subset $\Subset \subseteq \Int{n}$ such that
$\bldc \in \code\vert_{\overline{\Subset}}$. Hence,
\ifISIT
\else
\begin{eqnarray*}
\Height_m(\code)
& = &
\Height_m(\bldc) = \| \bldc \|_\infty \\
& \le &
\max_\Subset
\max_{\bldx \in \code\vert_{\overline{\Subset}}} \| \bldx \|_\infty .
\end{eqnarray*}
\fi

Next, we show that $\Height_m(\code)$ is at least
the right-hand side of~(\ref{eq:LP1}).
Let $\Subset$ and $\bldc \in \code\vert_{\overline{\Subset}}$
be where the maximum in~(\ref{eq:LP1}) is
\ifISIT
\else
attained (since $\bldx \mapsto \| \bldx \|_\infty$
is continuous and is maximized in~(\ref{eq:LP1}) over a compact set,
the maximum indeed exists).
\fi
Also, let $\pi(\cdot)$ be a sorting permutation of $\bldc$.
Since $\bldc \in \code\vert_{\overline{\Subset}}$,
there are no more than $m$ entries in $\bldc$ of magnitude
greater than $1$; in particular, $|c_{\pi(m)}| \le 1$.
Hence,
    \ifISIT
\else
\begin{eqnarray*}
\Height_m(\code)
& \ge &
\Height_m(\bldc) = \left| \frac{c_{\pi(0)}}{c_{\pi(m)}} \right|
\ge |c_{\pi(0)}| = \| \bldc \|_\infty \\
& = &
\max_\Subset
\max_{\bldx \in \code\vert_{\overline{\Subset}}} \| \bldx \|_\infty .
\end{eqnarray*}
\fi
\end{proof}

\ifISIT
\else
\begin{remark}
\label{rem:d>=m}
Theorem~\ref{thm:LP} holds also when $m \ge d$.
By~(\ref{eq:height-distance}),
the left-hand side of~(\ref{eq:LP1}) equals $\infty$ in this case.
As for the right-hand side, let $\bldc$ be any minimum-weight
codeword of $\code$ and let $\Subset$ be any $m$-subset of
$\Int{n}$ that is a superset of $\Support(\bldc)$.
Since all the scalar multiples of $\bldc$
belong to $\code\vert_{\overline{\Subset}}$, the inner maximum
in~(\ref{eq:LP1}) for this $\Subset$ is unbounded.\qed
\end{remark}

\subsection{Primal linear program}
\label{sec:LP-primal}
\fi

Theorem~\ref{thm:LP} implies
the algorithm in Figure~\ref{fig:LP} for computing the $m$-height of
a linear $[n,k,d]$ code $\code$ which is specified
by its $k \times n$ generator matrix
$G = (\bldg_0 \; \bldg_1 \; \ldots \; \bldg_{n-1})$.
The algorithm enumerates over
all $m$-subsets $\Subset \subseteq \Int{n}$ and over all indexes
$i \in \Subset$. For each pair $(\Subset,i)$,
the following linear program computes
a codeword $\bldc = (c_j)_{j \in \Int{n}} = \bldu G$
in $\code\vert_{\overline{\Subset}}$ for which the entry $c_i$
is maximized:
\begin{equation}
\label{eq:LP2}
\fbox{%
$%
\renewcommand{\arraystretch}{1.2}
\begin{array}{ll}
\textrm{Maximize} & \\
&
\bldu \cdot \bldg_i \\
\multicolumn{2}{l}{%
\textrm{over $\bldu \in \Realfield^k$, subject to}} \\
&
\left\{
\renewcommand{\arraystretch}{1}
\begin{array}{lcr}
\bldu \, (G)_{\overline{\Subset}} & \le & \bldone \\
\bldu \, (G)_{\overline{\Subset}} & \ge & -\bldone
\end{array}
\right.
\end{array}
$}
\end{equation}
(the notation $\bldone$ in~(\ref{eq:LP2}) stands
for the all-$1$ row vector in $\Realfield^{n-m}$,
and the inequalities therein are componentwise).
The return value of the algorithm is the largest $c_i$ found
across all pairs $(\Subset,i)$.

\begin{figure}[!hbt]
\centering
\begin{myalgorithm}
\noindent
\textbf{Input:}
\begin{itemize}
\item
$k \times n$ generator matrix
$G = (\bldg_0 \; \bldg_1 \; \ldots \; \bldg_{n-1})$ of $\code$
\item
$m \in \Int{1:d}$
\end{itemize}

\vspace{1ex}

\noindent
$h \leftarrow 0$

\noindent
\textbf{for} every $m$-subset $\Subset \subseteq \Int{n}$ \textbf{do}

\noindent\hspace*{2em}
\textbf{for} $i \in \Subset$ \textbf{do}

\noindent\hspace*{4em}
$c_i \leftarrow \max_{\bldu} \bldu \cdot \bldg_i$
as computed by~(\ref{eq:LP2})

\noindent\hspace*{4em}
$h \leftarrow \max \{ h, c_i \}$

\noindent
\textbf{Output:}
$\Height_m(\code) = h$

\end{myalgorithm}
\caption{Computation of the $m$-height of a linear $[n,k]$ code $\code$
using linear programming.}
\label{fig:LP}
\end{figure}

The algorithm requires solving $m \cdot \binom{n}{m}$ linear programs.
It can be seen as a simplified (and more efficient) version of
the algorithm described in~\cite[\S III.B]{Jiang}.

\ifISIT
\else
\begin{remark}
\label{rem:varyinggenerators}
The algorithm in Figure~\ref{fig:LP} works also when we use
a different generator matrix for each pair $(\Subset,i)$.
Moreover, it will work also if we permute the coordinates of $\code$,
as long as we keep track of the positions that are indexed
by $\Subset$ and $i$ after the permutation.
Now, when $m < d$ we have $n-m \ge n-d+1 \ge k$,
which means that the matrix $(G)_{\overline{\Subset}}$
in~(\ref{eq:LP2}) has full rank $k$ and, so, there is
a $k$-subset $\II \subseteq \overline{\Subset}$ such that
$(G)_\II^{-1} (G)_{\overline{\Subset}}$ contains a $k \times k$
identity matrix. Using $(G)_\II^{-1} G P$ in the inner loop
in Figure~\ref{fig:LP}, where $P$ is an $n \times n$ permutation matrix
that maps $\II$ to $\Int{k}$, results in a modified linear program
where $2k$ out of the $2(n-m)$ constraints in~(\ref{eq:LP2}) become
independent of $(\Subset,i)$ and take the form
\[
\phantom{\begin{array}{l} \\ \qed \end{array}}
\makebox[12.3ex]{}
\left\{
\begin{array}{lcr}
(\bldu)_{\Int{k}}  & \le & \bldone\phantom{.} \\
(\bldu)_{\Int{k}}  & \ge & -\bldone .
\end{array}
\right.
\makebox[12.3ex]{}
\begin{array}{r} \\ \qed \end{array}
\]
\end{remark}

\begin{remark}
\label{rem:LP-augmented}
The linear program~(\ref{eq:LP2}) involves $k$ variables that
satisfy $2(n-m)$ linear inequality constraints.
Its augmented form (as in~\cite[p.~11]{Luenberger}),
which is suitable for the simplex method,
involves $2(n-m)$ nonnegative slack variables that satisfy
$2(n-m)-k = n+r-2m$ linear equality constraints,
where $r = n - k$ is the redundancy of $\code$
(the vector $\bldu$ in~(\ref{eq:LP2})
can be expressed in terms of the slack variables and
therefore can be eliminated from the augmented form).\qed
\end{remark}

\subsection{Dual linear program}
\label{sec:LP-dual}
\fi

An alternate characterization of the $m$-height of a linear code can be
obtained by considering the dual of the linear program~(\ref{eq:LP2})
(see~\cite[Ch.~6]{BS}) and~\cite[Ch.~4]{Luenberger}).

\begin{theorem}
\label{thm:LP-dual}
Let $\code$ be a linear $[n,k,d]$ code over $\Realfield$,
let $G$ be a $k \times n$ generator matrix of $\code$ over $\Realfield$,
and let $m \in \Int{1:d}$. Then
\begin{equation}
\label{eq:LP-dual1}
\Height_m(\code) =
\max_\Subset
\max_{i \in \Subset}
\min_{\blde \in \EE_{\Subset,i}} \| \blde \|_1 ,
\end{equation}
where $\Subset$ is as in Theorem~\ref{thm:LP}
\ifISIT
\else
and\,\footnote{%
The affine space $\EE_{\Subset,i}$ depends only on $\code$, $\Subset$,
and $i$; in particular, it does not depend on the selected
generator matrix $G$.}
\fi
\begin{equation}
\label{eq:EE}
\EE_{\Subset,i}
= \left\{
\blde \in \Realfield^{n-m} \,:\,
(G)_{\overline{\Subset}} \, \blde = \bldg_i
\right\} .
\end{equation}
\end{theorem}

\begin{proof}
The dual of the linear program~(\ref{eq:LP2})
(which results in the same output value) takes the form:
\begin{equation}
\label{eq:LP-dual2}
\fbox{%
$%
\renewcommand{\arraystretch}{1.2}
\begin{array}{ll}
\textrm{Minimize} & \\
&
\bldone \cdot (\bldy + \bldz) \\
\multicolumn{2}{l}{%
\textrm{over $\bldy, \bldz \in \Realfield^{n-m}$, subject to}} \\
&
\left\{
\renewcommand{\arraystretch}{1}
\begin{array}{rcl}
(G)_{\overline{\Subset}} \, (\bldy - \bldz) & = & \bldg_i \\
\bldy & \ge & \bldzero \\
\bldz & \ge & \bldzero ,
\end{array}
\right.
\end{array}
$}
\end{equation}
where $\bldy$ and $\bldz$ are two column vectors of variables
that correspond, respectively, to the two sets of inequalities
in~(\ref{eq:LP2}).
The set of feasible solutions consists of all pairs of
nonnegative vectors $\bldy$ and $\bldz$
whose difference, $\bldy - \bldz$, is
a vector $\blde$ in $\EE_{\Subset,i}$.
Fixing any such vector $\blde$, the expression
$\bldone \cdot (\bldy + \bldz)$ is minimized when
\[
\bldy = \max \{ \blde, \bldzero \}
\quad \textrm{and} \quad
\bldz = -\min \{ \blde, \bldzero \} ,
\]
where the maximum and minimum are applied componentwise.
For this choice we get
$\bldone \cdot (\bldy + \bldz) = \| \blde \|_1$.
We conclude that
the solution to the linear program~(\ref{eq:LP-dual2}) is
the minimum of $\| \blde \|_1$ taken over
$\blde \in \EE_{\Subset,i}$. The outer enumeration over
$\Subset$ and $i$, as in Figure~\ref{fig:LP},
in turn, yields~(\ref{eq:LP-dual1}).
\end{proof}

\ifISIT
\else
The $k \times (n{-}m)$ matrix $(G)_{\overline{\Subset}}$ is
a generator matrix of the punctured code
$(\code)_{\overline{\Subset}}$.
As pointed out in Remark~\ref{rem:varyinggenerators},
when $m < d$ we have $\rank \left( (G)_{\overline{\Subset}} \right) = k$
and, so, $(\code)_{\overline{\Subset}}$ is
a linear $[n{-}m,k]$ code.\footnote{%
Conversely, when $m \ge d$, for any $m$-subset
$\Subset \subseteq \Int{n}$ that is a superset of the support of
a minimum-weight codeword of $\code$ we have
$\rank \left( (G)_{\overline{\Subset}} \right) < k$.
Yet $\rank(G) = k$ and, so, there is an index $i \in \Subset$
such that $\bldg_i$ is not in the linear span of the columns of
$(G)_{\overline{\Subset}}$. In this case
$\EE_{\Subset,i}$ is empty, which translates into
an infinite value in the minimum in~(\ref{eq:LP-dual1}).}
In view of this,
the affine space $\EE_{\Subset,i}$ in~(\ref{eq:EE}) is a coset of
the dual code $((\code)_{\overline{\Subset}})^\perp$,
which is a linear $[n{-}m,r{-}m]$ code
with $(G)_{\overline{\Subset}}$ as its parity-check matrix,
where $r = n-k$ is the redundancy of $\code$.
Specifically, $\EE_{\Subset,i}$
consists of all vectors $\blde \in \Realfield^{n-m}$ whose
syndrome, with respect to the parity-check matrix
$(G)_{\overline{\Subset}}$, equals $\bldg_i$.
The inner minimum in~(\ref{eq:LP-dual1}) is
the smallest $L_1$ norm of any vector in this coset,
i.e., the minimization is, in effect, a maximum-likelihood decoder
for a memoryless channel where
the error values obey a Laplace distribution.
\fi

Letting $B$ be an $(r{-}m) \times (n{-}m)$
parity-check matrix of $(\code)_{\overline{\Subset}}$
(and a generator matrix of $((\code)_{\overline{\Subset}})^\perp$),
the inner minimization in~(\ref{eq:LP-dual1}) can
\ifISIT
\else
thus
\fi
be recast in the form
\begin{equation}
\label{eq:LAD1}
\min_{\bldv \in \Realfield^{r-m}} \| \blda - \bldv B \|_1 ,
\end{equation}
for some given vector $\blda \in \Realfield^{n-m}$, where
both $B$ and $\blda$ depend on $\code$, $\Subset$, and $i$.
Specifically, $\blda$ is any vector in the coset $\EE_{\Subset,i}$,
and $\bldv$ in~(\ref{eq:LAD1}) is related to
$\blde$ in~(\ref{eq:LP-dual1}) by
\[
\blde = \blda - \bldv B .
\]
The minimization problem in~(\ref{eq:LAD1}), in turn, is known
as \emph{least absolute deviations (LAD) regression}~\cite{BS}.
It is a convex optimization problem in $r-m$ variables
and can be solved, \emph{inter alia}, using linear programming,
by introducing a vector $\widehat{\blde}$ of
$n-m$ auxiliary variables (see~\cite[p.~158]{BS}):
\begin{equation}
\ifISIT
\else
\label{eq:LAD2}
\fi
\fbox{%
$%
\renewcommand{\arraystretch}{1.2}
\begin{array}{ll}
\textrm{Minimize} & \\
&
\bldone \cdot \widehat{\blde} \\
\multicolumn{2}{l}{%
\textrm{over
$(\bldv,\widehat{\blde}) \in \Realfield^{r-m} \times \Realfield^{n-m}$,
subject to}} \\
&
\left\{
\renewcommand{\arraystretch}{1}
\begin{array}{rcr}
\bldv B + \widehat{\blde} & \ge & \blda \phantom{.} \\
\bldv B - \widehat{\blde} & \le & \blda .
\end{array}
\right.
\end{array}
$}
\end{equation}
\ifISIT
\else
Note that $\widehat{\blde}$ is related to $\bldy$ and $\bldz$
in~(\ref{eq:LP-dual2}) by
$\widehat{\blde} = \bldy + \bldz$ and its entries
are the magnitudes of the entries of $\blde$ in~(\ref{eq:LP-dual1}).

\begin{remark}
The linear program~(\ref{eq:LAD2}) involves $n+r-2m$ variables
that satisfy $2(n-m)$ linear inequality constraints,
and its augmented form involves $2(n-m)$ nonnegative slack variables
that satisfy $2(n-m) - (n+r-2m) = k$ linear equality constraints
(the same as the linear program~(\ref{eq:LP-dual2})).\qed
\end{remark}
\fi

We end this section by presenting a structural property of
$m$-extremal codewords.

\begin{lemma}
\label{lem:structural}
Let $\code$ be a linear $[n,k,d]$ code over $\Realfield$
with a dual code of minimum
\ifISIT
\else
distance\footnote{%
The case $d^\perp = 1$ corresponds to a code $\code$ that has
a coordinate which is identically $0$ across all codewords.
Shortening $\code$ on that coordinate does not affect its $m$-height.}
\fi
$d^\perp \ge 2$.
For $m \in \Int{1:d}$, let $\bldc = (c_j)_{j \in \Int{n}}$ be
any $m$-extremal codeword in $\code$
and let $\pi(\cdot)$ be a sorting permutation of $\bldc$.
Then,
\[
|c_{\pi(m)}| = |c_{\pi(m+1)}| = \ldots = |c_{\pi(m+d^\perp-2)}| \ne 0
\]
(i.e., $\PP_m(\bldc) \supseteq \Int{m:m{+}d^\perp{-}1}$ and,
so, $|\PP_m(\bldc)| \ge d^\perp-1$).
\end{lemma}

\begin{proof}
Since $m < d$ we have $c_{\pi(m)} \ne 0$;
without loss of generality we assume hereafter
that $|c_{\pi(m)}| = 1$.
Fix a $k \times n$ generator matrix $G$ of $\code$
and consider the linear program~(\ref{eq:LP2})
when $i \leftarrow \pi(0)$ and
\ifISIT
\else
\[
\Subset \leftarrow \left\{ \pi(t) \,:\, t \in \Int{m} \right\} .
\]
\fi
The maximum then is attained
at the vector $\bldu$ that (uniquely) satisfies $\bldc = \bldu G$.
The result will follow once we show that
there exists a $(d^\perp{-}1)$-subset $\II$ of $\overline{\Subset}$
such that, for this maximizing $\bldu$,
\[
(\bldc)_\II = \bldu \, (G)_\II \in \{ \pm 1 \}^{d^\perp-1} .
\]

Suppose to the contrary that such a subset $\II$ does not exist, namely,
there exists a subset $\JJ \subseteq \overline{\Subset}$ of
size $|\JJ| > |\overline{\Subset}| - d^\perp + 1$ such that
at any coordinate $j \in \JJ$ of $\bldu \, (G)_{\overline{\Subset}}$,
both inequalities in~(\ref{eq:LP2}) are strict.
Referring to the proof of Theorem~\ref{thm:LP-dual},
consider the minimizing vectors $\bldy$ and $\bldz$ of
the dual linear program~(\ref{eq:LP-dual2}).
By the complementary slackness theorem~\cite[p.~96]{Luenberger},
both $\bldy$ and $\bldz$ have zeroes at all
the positions that are indexed by $\JJ$.
Consequently, $\blde = \bldy - \bldz$ has zeroes at
those positions as well, which means that
\begin{equation}
\label{eq:v1}
\weight(\blde) \le |\overline{\Subset}| - |\JJ| < d^\perp - 1 .
\end{equation}
On the other hand, by~(\ref{eq:LP-dual2}) we also have
\begin{equation}
\label{eq:v2}
(G)_{\overline{\Subset}} \, \blde = \bldg_i .
\end{equation}
Combining~(\ref{eq:v1}) and~(\ref{eq:v2}) we conclude that
there are $d^\perp - 1$ linearly dependent columns in $G$.
However, this is impossible, since $G$ is a parity-check matrix of
the dual code $\code^\perp$ and, as such,
every $d^\perp - 1$ columns in $G$ are linearly independent.
\end{proof}

\section{Combinatorial characterization of the height}
\label{sec:combinatorial}

In this section, we present characterizations of the $m$-height of
a linear code as a maximum (or a maximum--minimum)
taken over an explicit subset of finitely many codewords.
In essence, these codewords
correspond to a subset of the vertices of
the union of the feasible regions of
the linear program~(\ref{eq:LP2})
\ifISIT
\else
(respectively, the dual linear program~(\ref{eq:LP-dual2}))
\fi
as  $\Subset$ ranges over all the $m$-subsets of $\Int{n}$
and $i$ ranges over the elements in each $\Subset$.

\ifISIT
\else
\subsection{Primal characterization}
\label{sec:combinatorial-primal}
\fi

The next lemma adds to Lemma~\ref{lem:structural} in that it implies
that there is always an $m$-extremal codeword $\bldc \in \code$
such that $|\PP_m(\bldc)| \ge k$.
\ifISIT
\else
(Note that for the special case of MDS codes, we get from
Lemma~\ref{lem:structural} that $|\PP_m(\bldc)| \ge k$
for \underline{any} $m$-extremal codeword $\bldc$ in $\code$.)
\fi

\begin{lemma}
\label{lem:m-height}
Let $\code$ be a linear $[n,k,d]$ code over $\Realfield$
and let $m \in \Int{1:d}$.
Suppose that $\bldc$ is an $m$-extremal codeword in $\code$
with a largest set $\TT = \PP_m(\bldc)$
among all $m$-extremal codewords in $\code$. Then for
any $k \times n$ generator matrix $G$ of $\code$ over $\Realfield$,
\[
\rank \left( (G)_\TT \right) = k
\]
(in particular, $|\TT| \ge k$).

Equivalently, in any parity-check matrix of $\code$ over $\Realfield$,
the columns in $(H)_{\overline{\TT}}$ are linearly independent.
\end{lemma}

\ifISIT
\else
\begin{proof}
Let $\pi(\cdot)$ be a sorting permutation of $\bldc$; we assume
hereafter in the proof (without loss of generality)
that $|c_{\pi(m)}| = 1$.

Suppose to the contrary that
$\rank \left( (G)_\TT \right) < k$
(or that for a parity-check matrix $H$ of $\code$,
the columns in $(H)_{\overline{\TT}}$
are linearly dependent). Then there exists a nonzero codeword
$\bldx = (x_j)_{j \in \Int{n}} \in \code$ such that
$\Support(\bldx) \subseteq \overline{\TT}$, i.e.,
\begin{equation}
\label{eq:disjoint}
\Support(\bldx) \cap \TT = \emptyset .
\end{equation}
For $\varepsilon \in \Realfield$, define
\[
\bldc_\varepsilon
= (c_{\varepsilon,j})_{j \in \Int{n}}
= \bldc + \varepsilon \cdot \bldx .
\]
We distinguish between two cases.

\emph{Case 1:}
$\Support(\bldx) \cap \PP_0(\bldc) \ne \emptyset$.
Let $i \in \Support(\bldx) \cap \PP_0(\bldc)$;
by possibly negating $\bldx$
we can assume that $c_{\pi(i)}$ and $x_{\pi(i)}$ have the same sign.
Let $\varepsilon \in \Realfield$ be in the range
\[
0 < \varepsilon < \min_{j \in \Support(\bldx)}
\Bigl| \frac{c_j}{x_j} \Bigr| .
\]
By~(\ref{eq:disjoint}), for every $j \in \Int{n}$ we have:
\[
\begin{array}{ccc}
|c_j| = 1 & \Longleftrightarrow & |c_{\varepsilon,j}| = 1, \\
|c_j| > 1 & \Longleftrightarrow & |c_{\varepsilon,j}| > 1, \\
|c_j| < 1 & \Longleftrightarrow & |c_{\varepsilon,j}| < 1.
\end{array}
\]
Hence,
\[
\PP_m(\bldc_\varepsilon) = \TT
\]
(and all the entries in $(\bldc_\varepsilon)_\TT$ are $\pm 1$).
On the other hand,
\[
|c_{\pi(i)} + \varepsilon \cdot x_{\pi(i)}|
> |c_{\pi(i)}| = |c_{\pi(0)}| ,
\]
which implies that $\Height_m(\bldc_\varepsilon) > \Height_m(\bldc)$,
thereby contradicting the $m$-extremality of $\bldc$.

\emph{Case 2:}
$\Support(\bldx) \cap \PP_0(\bldc) = \emptyset$.
Let $\varepsilon$ be of smallest magnitude such that
$(\bldc_\varepsilon)_{\overline{\TT}}$
contains an entry of magnitude $1$;
explicitly,
\begin{equation}
\label{eq:epsilon}
\varepsilon = \pm \min_{j \in \Support(\bldx)}
\left| \frac{1 - |c_j|}{x_j} \right| ,
\end{equation}
with the sign determined to match that of
$c_{j'} \cdot (1 - |c_{j'}|)/x_{j'}$, where $j'$ is an index $j$
at which the minimum in~(\ref{eq:epsilon}) is attained.

On the one hand, for every $j \in \Int{n}$:
\[
\begin{array}{ccc}
|c_j| = 1 & \Longrightarrow & |c_{\varepsilon,j}| = 1,   \\
|c_j| > 1 & \Longrightarrow & |c_{\varepsilon,j}| \ge 1, \\
|c_j| < 1 & \Longrightarrow & |c_{\varepsilon,j}| \le 1.
\end{array}
\]
These, in turn, imply that
\[
\TT \subseteq \PP_m(\bldc_\varepsilon) .
\]
Moreover, the containment is strict, since there exists at least
one index, namely, $j' \in \overline{\TT}$,
where $|c_{\varepsilon,j'}| = 1$ while $|c_{j'}| \ne 1$.

On the other hand,
\[
\Height_m(\bldc_\varepsilon)
= |c_{\varepsilon,\pi(0)}| = |c_{\pi(0)}| = \Height_m(\bldc) .
\]
Thus, $\bldc_\varepsilon$ is $m$-extremal
and $|\PP_m(\bldc_\varepsilon)| > |\TT|$;
yet this contradicts the assumption on the maximality of $|\TT|$
among all $m$-extremal codewords in $\code$.
\end{proof}

\begin{remark}
\label{rem:m-height-m>=d}
Lemma~\ref{lem:m-height} holds also when $m \ge d$,
except when the code $\code$ therein is MDS.
Specifically, when $m \ge d$,
for any $m$-extremal codeword $\bldc \in \code$
(in particular, for any minimum-weight codeword of $\code$)
we have $c_{\pi(m)} = 0$ (and therefore $\Height_m(\bldc) = \infty)$,
yet when $\code$ is MDS,
none of its nonzero codewords has more than $k-1$ zero entries.\qed
\end{remark}
\fi

Lemma~\ref{lem:m-height} implies the next theorem, which presents
another characterization of the $m$-height of a linear code.
We introduce the following notation.
For a linear $[n,k{=}n{-}r]$ code $\code$, we let
$\Upsilon(\code)$ be the set of all $k$-subsets $\II \subseteq \Int{n}$
such that
\ifISIT
\else
the projections of the codewords of $\code$ on $\II$
range over the whole space $\Realfield^k$, namely:
\[
\bigl\{ (\bldc)_\II \,:\, \bldc \in \code \bigr\} = \Realfield^k .
\]
Equivalently,
\fi
in any $k \times n$ generator matrix $G$ of $\code$,
the submatrix $(G)_\II$ is nonsingular.
\ifISIT
\else
It is easy to see that a $k$-subset $\II$ belongs to $\Upsilon(\code)$
if and only if its complement $r$-subset
$\overline{\II}$ belongs to $\Upsilon(\code^\perp)$.
\fi

\begin{theorem}
\label{thm:m-height}
Let $\code$ be a linear $[n,k,d]$ code over $\Realfield$,
let $m \in \Int{1:d}$,
and let $G$ be any $k \times n$ generator matrix of $\code$
over $\Realfield$.
Then
\begin{equation}
\ifISIT
\else
\label{eq:m-height1}
\fi
\Height_m(\code) =
\max_{\II \in \Upsilon(\code)}
\max_{\bldu \in \{ \pm 1 \}^k}
\Height_m \bigl( \bldu \, (G)_\II^{-1} G \bigr) .
\end{equation}
\ifISIT
\else

Also,
\begin{equation}
\label{eq:m-height2}
\Height_m(\code) =
\max_{\II \in \Upsilon(\code)}
\max_\bldu
\,
\bigl\| \bldu \, (G)_\II^{-1} (G)_{\overline{\II}} \bigr\|_\infty  ,
\end{equation}
where $\bldu$ ranges over all the elements in $\{ \pm 1 \}^k$
such that the vector
$\bldu \, (G)_\II^{-1} (G)_{\overline{\II}} \; (\in \Realfield^{n-k})$
has no more than $m$ entries of magnitude greater than $1$
and no more than $n-m-1$ entries of magnitude smaller than $1$.
\fi
\end{theorem}

\ifISIT
\else
\begin{proof}
Define the following subset $\code'$ of $\code$:
\[
\code'
=
\Bigl\{
\bldc \in \code \,:\,
(\bldc)_\II \in \{ \pm 1 \}^k \;
\textrm{for some} \; \II \in \Upsilon(\code) \Bigr\}
\]
(the subset $\II$ may depend on $\bldc$).
The right-hand side of~(\ref{eq:m-height1}) then equals
\[
\max_{\bldc \in \code'} \Height_m(\bldc) .
\]
Eq.~(\ref{eq:m-height1})
now follows from Lemma~\ref{lem:m-height}, which implies
that $\code'$ contains at least one $m$-extremal codeword of $\code$.

Turning to Eq.~(\ref{eq:m-height2}),
given some $\II \in \Upsilon(\code)$, consider a codeword
$\bldc = \bldu \, (G)_\II^{-1} G$.
The assumed conditions on $\bldu$
in the inner maximization in~(\ref{eq:m-height2})
are then equivalent to requiring that
$(\bldc)_\II \in \{ \pm 1 \}^k$ and
$\II \subseteq \PP_m(\bldc)$.
Hence, the right-hand side of (\ref{eq:m-height2}) equals
\[
\max_{\bldc \in \code''} \| \bldc \|_\infty ,
\]
where
\begin{eqnarray*}
\lefteqn{
\code''
=
\Bigl\{
\bldc \in \code \,:\,
(\bldc)_\II \in \{ \pm 1 \}^k \;
\textrm{for some} \; \II \in \Upsilon(\code)
} \makebox[25ex]{} \\
&&
\textrm{such that} \; \II \subseteq \PP_m(\bldc)
\Bigr\} .
\end{eqnarray*}
Now, for $\bldc \in \code''$ we have
$\Height_m(\bldc) = \| \bldc \|_\infty$, which means
that the right-hand side of~(\ref{eq:m-height2}) also equals
\[
\max_{\bldc \in \code''} \Height_m(\bldc) .
\]
Lemma~\ref{lem:m-height} again implies
that $\code''$ contains at least one $m$-extremal codeword of $\code$.
\end{proof}

Noting that the vectors $\bldu$ in the inner maximizations
in both~(\ref{eq:m-height1}) and~(\ref{eq:m-height2})
can be restricted to have $1$ as their first entry,
the total number of terms to maximize over
in the right-hand side of~(\ref{eq:m-height1}) (or~(\ref{eq:m-height2}))
is at most
\begin{equation}
\label{eq:complexity}
\binom{n}{k} \cdot 2^{k-1} .
\end{equation}

\begin{remark}
\label{rem:m-height}
Eq.~(\ref{eq:m-height2}) can be expressed also in terms of
a full-rank $r \times n$ parity-check matrix $H$ of $\code$
(where $r = n-k$). Recall that for every $r$-subset
$\JJ \in \Upsilon(\code^\perp)$, its complement $k$-subset
$\II = \overline{\JJ}$ belongs to $\Upsilon(\code)$ (and vice versa),
and we have
\[
(H)_\JJ^{-1} (H)_{\overline{\JJ}}
=
-\left( (G)_\II^{-1} (G)_{\overline{\II}} \right)^\top .
\]
    From~(\ref{eq:m-height2}) we then get
\begin{equation}
\label{eq:m-height3}
\Height_m(\code) =
\max_{\JJ \in \Upsilon(\code^\perp)}
\max_\bldu
\,
\bigl\| (H)_\JJ^{-1} (H)_{\overline{\JJ}} \,\bldu^\top \bigr\|_\infty  ,
\end{equation}
where $\bldu$ ranges over all the elements in $\{ \pm 1 \}^k$
such that $(H)_\JJ^{-1} (H)_{\overline{\JJ}} \, \bldu^\top$
has no more than $m$ entries of magnitude greater than $1$
and no more than $n-m-1$ entries of magnitude smaller than $1$.
In both~(\ref{eq:m-height2}) and~(\ref{eq:m-height3}),
the total number of terms to maximize over is the same
(and the range of $\bldu$ is the same).
However, the computation of~(\ref{eq:m-height3}) involves inverting
$r \times r$ matrices, compared to~(\ref{eq:m-height2}) where
the order of the inverted matrices is $k \times k$;
hence, from a complexity viewpoint, Eq.~(\ref{eq:m-height3})
may be preferable when $k > r$.\qed
\end{remark}

\subsection{Dual characterization}
\label{sec:combinatorial-dual}
\fi

We next turn to results that are conceptual duals of
Lemma~\ref{lem:m-height} and Theorem~\ref{thm:m-height}, based on
the characterization~(\ref{eq:LP-dual1}) of
the $m$-height of a linear code.
\ifISIT
\else

When $m = r$ we have $|\EE_{\Subset,i}| = 1$, in which case
the minimization in~(\ref{eq:LP-dual1}) becomes trivial
(we will deal with this case in Section~\ref{sec:r-height}).

The next lemma provides an explicit solution to
the minimization problem~(\ref{eq:LAD1})
(and hence to~(\ref{eq:LP-dual1})) for the case $m = r-1$.
This lemma is found in~\cite[\S 1.2]{BS} and seems to have its roots
in the 18th century. For completeness, we include
a proof in the Appendix.

\begin{lemma}
\label{lem:m=r-1}
Let $\blda = (a_j)_{j \in \Int{\ell}}$
and $\bldb = (b_j)_{j \in \Int{\ell}}$
be two vectors in $\Realfield^\ell$ where
$\bldb \ne \bldzero$.
The function $f : \Realfield \rightarrow \Realfield$
that is defined by
\[
v \mapsto f(v) = \| \blda - v \cdot \bldb \|_1
\]
attains its minimum (not necessarily uniquely) at
the weighted median, $v_{\mathrm{med}}$, among the values
\[
\makebox[6ex]{$a_j/b_j,$} \quad j \in \Support(\bldb) ,
\]
with respective weights
\[
\makebox[6ex]{$|b_j|,$} \quad j \in \Support(\bldb) .
\]
In particular, at that minimum, the vector
$\blda - v_{\mathrm{med}} \cdot \bldb$ has at least one zero coordinate
at some position in $\Support(\bldb)$.
\end{lemma}

\fi
The following lemma is the dual counterpart of Lemma~\ref{lem:m-height}.

\begin{lemma}
\label{lem:m-height-dual}
Using the notation of Theorem~\ref{thm:LP-dual},
given an $m$-subset $\Subset \subseteq \Int{n}$
and $i \in \Subset$, let $\blde'$ be a vector in $\EE_{\Subset,i}$
that satisfies the following two conditions.
\begin{list}{}{\settowidth{\labelwidth}{\textrm{(ii)}}}
\item[(i)]
It attains the inner minimum
\ifISIT
\else
in~(\ref{eq:LP-dual1}), i.e.,
\begin{equation}
\label{eq:LP-dual1again}
\| \blde' \|_1 = \min_{\blde \in \EE_{\Subset,i}} \| \blde \|_1 .
\end{equation}
\fi
\item[(ii)]
It has the smallest support, $\TT = \Support(\blde')$,
among all vectors in $\EE_{\Subset,i}$ that attain that minimum.
\end{list}
Then the columns in $(G)_\TT$ are linearly independent
(in particular, $|\TT| \le k$).

Equivalently, writing $\overline{\TT} = \Int{n{-}m} \setminus \TT$,
for any $(r{-}m) \times (n{-}m)$ parity-check matrix $B$ of
the punctured code $(\code)_{\overline{\Subset}}$,
\[
\rank \left( (B)_{\overline{\TT}} \right) = r-m .
\]
\end{lemma}

The lemma follows from
the fundamental theorem of linear programming, which states
that the minimum in the linear program~(\ref{eq:LP-dual2})
is attained
\ifISIT
\else
by a so-called basic feasible solution, namely,
\fi
by vectors $\bldy$ and $\bldz$ such that
the columns of $(G)_{\overline{\Subset}}$ that are indexed by
$\Support(\bldy) \cup \Support(\bldz)$ are linearly independent
(see~\cite[p.~19]{Luenberger}).
\ifISIT
\else
We next present
an alternate proof which is based on Lemma~\ref{lem:m=r-1};
our proof is a modification of the proof of
Theorem~1 in~\cite[\S 1.2]{BS}.

\begin{proof}[Proof of Lemma~\ref{lem:m-height-dual}]
Suppose to the contrary that
the columns in $(G)_\TT$ are linearly dependent
(or that $\rank \left( (B)_{\overline{\TT}} \right) < r-m$).
Then there exists a nonzero codeword
$\bldx \in ((\code)_{\overline{\Subset}})^\perp$
with support $\Support(\bldx) \subseteq \TT$.
Consider the vectors
\[
\blde' - \varepsilon \cdot \bldx ,
\quad
\varepsilon \in \Realfield .
\]
Their supports are all subsets of $\TT$;
moreover, these vectors are all in the same coset, $\EE_{\Subset,i}$, of
$((\code)_{\overline{\Subset}})^\perp$.
Therefore,
\begin{equation}
\label{eq:min}
\| \blde' \|_1
\stackrel{\textrm{(\ref{eq:LP-dual1again})}}{=}
\min_{\blde \in \EE_{\Subset,i}} \| \blde \|_1
\le
\min_{\varepsilon \in \Realfield}
\| \blde' - \varepsilon \cdot \bldx \|_1
\le \| \blde' \|_1 ,
\end{equation}
which means that the inequalities in~(\ref{eq:min})
are in fact equalities.
By Lemma~\ref{lem:m=r-1} it follows that
there is a minimizer $\varepsilon_{\min}$ in~(\ref{eq:min})
such that, at some position
$h \in \Support(\bldx) \; (\subseteq \TT)$,
the vector
$\blde'' = \blde' - \varepsilon_{\min} \cdot \bldx$ has
a zero entry. Hence,
\[
\Support(\blde'') \subseteq \TT \setminus \{ h \}
\subsetneq \TT ,
\]
namely, $\blde''$ attains the minima in~(\ref{eq:min})
and has a smaller support than $\blde'$.
This, however, contradicts condition~(ii).
\end{proof}
\fi

Lemma~\ref{lem:m-height-dual} implies the following
(dual) combinatorial characterization of
the $m$-height of a linear code.

\begin{theorem}
\label{thm:m-height-dual}
Let $\code$ be a linear $[n,k{=}n{-}r,d]$ code over $\Realfield$
and let $m \in \Int{1:d}$. Then the following holds.

\begin{list}{}{\settowidth{\labelwidth}{\textrm{(iii)}}}
\item[(i)]
For any $k \times n$ generator matrix
$G = (\bldg_0 \; \bldg_1 \; \ldots \; \bldg_{n-1})$ of $\code$
over $\Realfield$,
\begin{equation}
\label{eq:m-height-dual1}
\Height_m(\code)
=
\max_\Subset
\;
\max_{i \in \Subset}
\min_{\genfrac{}{}{0ex}{1}%
        {\II \in \Upsilon(\code) \,:}{\II \subseteq \overline{\Subset}}}
\;
\bigl\| (G)_\II^{-1} \bldg_i \bigr\|_1 ,
\end{equation}
where $\Subset$ ranges over all $m$-subsets of $\Int{n}$.
\item[(ii)]
For any $r \times n$ parity-check matrix $H$ of $\code$
over $\Realfield$,
\begin{equation}
\label{eq:m-height-dual2}
\Height_m(\code)
=
\max_\Subset
\;
\max_{i \in \Subset}
\min_{\genfrac{}{}{0ex}{1}%
          {\JJ \in \Upsilon(\code^\perp) \,:}{\JJ \supseteq \Subset}}
\;
\left\|
\left[ (H)_\JJ^{-1} \right]_i (H)_{\overline{\JJ}}
\right\|_1 ,
\end{equation}
where $\Subset$ is as in
\ifISIT
\else
part~(i).\footnote{%
We index the rows of $(H)_{\JJ}^{-1}$ by the $r$-subset $\JJ$,
which is the set of column indexes of $(H)_\JJ$.}
\fi
\end{list}
\end{theorem}

\begin{proof}
(i)~Let $\Subset \subseteq \Int{n}$ and $i \in \Subset$ be maximizers of
the outer maxima in~(\ref{eq:LP-dual1}).
By Lemma~\ref{lem:m-height-dual}, the inner minimum
in~(\ref{eq:LP-dual1}) is attained
at a vector $\blde' \in \EE_{\Subset,i}$
whose support, $\TT = \Support(\blde')$,
is such that the columns in $(G)_\TT$ are linearly independent.
Since $|\overline{\Subset}| = n-m \ge n-d+1 \ge k$,
the $k \times (n{-}m)$ matrix $(G)_{\overline{\Subset}}$
has full rank $k$ and, so, we can extend $\TT$ to
a $k$-subset $\II \subseteq \overline{\Subset}$
that belongs to $\Upsilon(\code)$. From
\[
\Support(\blde') \subseteq \II
\quad \textrm{and} \quad
(G)_{\overline{\Subset}} \, \blde'
\stackrel{\textrm{(\ref{eq:EE})}}{=} \bldg_i
\]
it follows that the vector $\blde'' = (G)_\II^{-1} \bldg_i$,
while being a subvector of $\blde'$, still has $\TT$ as its support.
Hence, $\| \blde'' \|_1 = \| \blde' \|_1$, which implies that
the right-hand sides of~(\ref{eq:LP-dual1})
and~(\ref{eq:m-height-dual1}) are equal.

\ifISIT
\else
(ii)~We repeat the argument in Remark~\ref{rem:m-height}: for
\fi
every $\JJ \in \Upsilon(\code^\perp)$,
the complement $\II = \overline{\JJ}$ belongs to $\Upsilon(\code)$, and
\[
(H)_\JJ^{-1} (H)_{\overline{\JJ}}
=
-\left( (G)_\II^{-1} (G)_{\overline{\II}} \right)^\top .
\]
Hence, for every row index $i \in \JJ \; (\supseteq \Subset)$ of
$(H)_\JJ^{-1}$,
\[
\left\|
\left[ (H)_\JJ^{-1} \right]_i (H)_{\overline{\JJ}}
\right\|_1
=
\bigl\| (G)_\II^{-1} \bldg_i \bigr\|_1 ,
\]
and the result follows from part~(i).
\end{proof}

\ifISIT
\else
The total number of terms to maximize/minimize over
in~(\ref{eq:m-height-dual1})--(\ref{eq:m-height-dual2})
is at most
\begin{equation}
\label{eq:complexity-dual}
\binom{n}{m} \cdot m \cdot \binom{n-m}{k}
=
m \cdot \binom{n}{r} \cdot \binom{r}{m} .
\end{equation}
This figure is smaller than~(\ref{eq:complexity})
for high-rate codes (when $k \ge r + \Omega(\log r)$)
or when $m$ is either small or close to $r$.
Comparing the computations of~(\ref{eq:m-height-dual1})
and~(\ref{eq:m-height-dual2}),
the latter is preferable when $k > r$ (see Remark~\ref{rem:m-height}).

\begin{remark}
\label{rem:m-height-dual-m>=d}
When $m \ge d$, the inner minimizations
in Eqs.~(\ref{eq:m-height-dual1})--(\ref{eq:m-height-dual2})
are over empty sets for any $m$-subset $\Subset$ that is
a superset of the support of a minimum-weight codeword of $\code$;
in this case the inner minimum can be seen
as infinity. In contrast, we note that
Eqs.~(\ref{eq:m-height1}), (\ref{eq:m-height2}),
and~(\ref{eq:m-height3}), yield finite values even when $m \ge d$,
so Theorem~\ref{thm:m-height} clearly does not hold in this case.\qed
\end{remark}

The next corollary identifies an $m$-extremal codeword
in case the code is MDS.

\begin{corollary}
\label{cor:m-height-MDS}
Using the notation in Theorem~\ref{thm:m-height-dual},
suppose that $\code$ therein is MDS
and let $\Subset$, $i$, and $\II$ (respectively,
$\JJ = \overline{\II}$) be
where the maximum--minimum in~(\ref{eq:m-height-dual1})
(respectively, (\ref{eq:m-height-dual2}) is attained.
Let $\bldc$ be the unique codeword in $\code$
that is defined by
\begin{equation}
\label{eq:m-height-MDS1}
\bldc = (\bldc)_\II (G)_\II^{-1} G ,
\end{equation}
where
\begin{equation}
\label{eq:m-height-MDS2}
(\bldc)_\II =
\sgn \,
\bigl( (G)_\II^{-1} \bldg_i \bigr)^\top
=
- \sgn \,
\bigl(
\bigl[ (H)_\JJ^{-1} \bigr]_i (H)_{\overline{\JJ}}
\bigr) ,
\end{equation}
with $\sgn(\cdot)$ applied componentwise.
Then $\bldc$ is $m$-extremal in $\code$.
\end{corollary}

\begin{proof}
As argued in the proof of Theorem~\ref{thm:m-height-dual},
the inner minimum in~(\ref{eq:LP-dual1}) is attained
at a vector $\blde' = (e'_j)_{j \in \overline{\Subset}}$
with support $\Support(\blde') \subseteq \II$.
We next claim the equality $\Support(\blde') = \II$.
Otherwise, from $(G)_{\overline{\Subset}} \, \blde' = \bldg_i$
we would have $k$ linearly dependent columns in $G$,
which is impossible since $\code$ is MDS.

The respective vectors
\[
\bldy' = (y'_j)_j = \max \{ \blde', \bldzero \}
\quad \textrm{and} \quad
\bldz' = (z'_j)_j = -\min \{ \blde', \bldzero \}
\]
(which satisfy $\bldy' - \bldz' = \blde'$) are minimizers of
the dual linear program~(\ref{eq:LP-dual2}).\footnote{%
In the terminology of linear programming, the MDS property,
which was used earlier, implies that
the dual linear program program~(\ref{eq:LP-dual2}) is nondegenerate.}
Letting $\bldu'$ be a maximizer of
the primal program~(\ref{eq:LP2}),
by the complementary slackness theorem
we have, for every $j \in \II$,
\[
\begin{array}{ccccl}
e'_j > 0 & \Longrightarrow & y'_j > 0
& \Longrightarrow & \bldu' \cdot \bldg_j = \phantom{-}1  \\
e'_j < 0 & \Longrightarrow & z'_j > 0
& \Longrightarrow & \bldu' \cdot \bldg_j = -1 .
\end{array}
\]
If follows that
\[
\bldu' (G)_\II = \sgn \, \bigl( (G)_\II^{-1} \bldg_i \bigr)^\top
\]
(which is a vector in $\{ \pm 1 \}^k$).
The vector $\bldc$
in~(\ref{eq:m-height-MDS1})--(\ref{eq:m-height-MDS2}),
in turn, equals $\bldu' G$ and is thus $m$-extremal in $\code$.
\end{proof}
\fi

\begin{example}
\label{ex:r=2}
We consider here the MDS construction from~\cite[\S V]{Roth1}.
Given an integer $n > 2$,
let $\alpha =  \pi/n$ and let $\omega$ denote
the complex primitive $(2n)$th root of unity $e^{\imath \alpha}$,
where $\imath = \sqrt{-1}$.
We define $\code(n)$ to be the following
linear $[n,n{-}2]$ code over $\Realfield$:
\[
\code(n) =
\left\{
\textstyle
\bldc = (c_j)_{j \in \Int{n}} \in \Realfield^n \,:\,
\sum_{j \in \Int{n}} c_j \omega^j = 0
\right\}
.
\]
This code is MDS and negacyclic
with a $2 \times n$ parity-check matrix
$H = (\bldh_j)_{j \in \Int{n}}$ whose columns are given by
\begin{equation}
\label{eq:hj}
\bldh_j =
\left(
\begin{array}{cc}
\cos (j \alpha) \\
\sin (j \alpha)
\end{array}
\right) ,
\quad
j \in \Int{n}
\end{equation}
\ifISIT
\else
(note that this matrix is different from the parity-check matrix
used in~\cite[\S V]{Roth1}).
\fi
Using purely geometric arguments, it was shown
in~\cite[Proposition~11]{Roth1} that
\begin{equation}
\label{eq:r=2,m=2}
\Height_2(\code(n))
= \frac{1}{2 \sin^2 (\alpha/2)} - 1
.
\end{equation}
We will re-derive this formula using
Theorem~\ref{thm:m-height-dual} and also show that
\begin{equation}
\label{eq:r=2,m=1}
\Height_1(\code(n)) =
\left\{
\ifISIT
\else
\renewcommand{\arraystretch}{1.5}
\fi
\begin{array}{lcl}
\ifISIT
\else
\displaystyle
\frac{1}{\sin (\alpha/2)} - 1,
\fi
                               && \textrm{when $n$ is odd} \\
\cot (\alpha/2) - 1,           && \textrm{when $n$ is even} .
\end{array}
\right.
\end{equation}

Due to the negacyclic property of $\code(n)$, we can assume hereafter
that the maximum in~(\ref{eq:m-height-dual2}) is attained
at $i = 0$ (and at a subset $\Subset$ that contains the element $0$).
Let $\JJ = \{ 0, t \}$, for some $t \in \Int{1:n}$.
Then
\[
(H)_\JJ^{-1}
=
\left(
\arraycolsep0.5ex
\begin{array}{c@{\;\;\;}c}
1 & \cos (t \alpha) \\
0 & \sin (t \alpha)
\end{array}
\right)^{\!-1}
{\!\!\!\!\!\!} =
\frac{1}{\sin (t \alpha)}
\left(
\arraycolsep0.5ex
\begin{array}{cc}
\sin (t \alpha) & -\!\cos (t \alpha) \\
0               & 1
\end{array}
\right)
\]
and, therefore, for every $j \in \Int{n}$,
\ifISIT
\else
\begin{eqnarray}
\bigl[ (H)_\JJ^{-1} \bigr]_0 \, \bldh_j
\!\! & = & \!
\frac{1}{\sin (t \alpha)}
\left(
\arraycolsep0.5ex
\begin{array}{cc}
\sin (t \alpha) & -\!\cos (t \alpha)
\end{array}
\right)
\left(
\arraycolsep0.5ex
\begin{array}{c}
\cos (j \alpha) \\ \sin (j \alpha)
\end{array}
\right) \nonumber \\
\label{eq:r=2}
& = & \!
\frac{\sin ((t{-}j) \alpha)}{\sin (t \alpha)} .
\end{eqnarray}
\fi
Now,
\begin{eqnarray*}
\lefteqn{
\sum_{j \in \Int{1:n} \setminus \{ t \}}
| \sin ((t{-}j) \alpha) |
=
\sum_{j \in \Int{1:n} \setminus \{ t \}} \sin (j \alpha)
} \makebox[2ex]{} \\
& = &
\Bigl( \sum_{j \in \Int{n}} \sin (j \alpha) \Bigr)
- \sin (t \alpha) = \cot (\alpha/2) - \sin (t \alpha) ,
\end{eqnarray*}
where the last step
\ifISIT
\else
follows from Lagrange's
trigonometric identities~\cite[\S 2.4.1.6, No.~13]{JeffreyDai}.
\fi
Combining this with~(\ref{eq:r=2}), we conclude that
\begin{equation}
\label{eq:r=2,t}
\left\|
\left[ (H)_\JJ^{-1} \right]_0 (H)_{\overline{\JJ}}
\right\|_1
=
\frac{\cot (\alpha/2)}{\sin (t \alpha)} - 1 .
\end{equation}

For $m = 2$, we substitute $\Subset = \JJ = \{ 0, t \}$
in~(\ref{eq:m-height-dual2}), which translates into
maximizing~(\ref{eq:r=2,t}) over $t \in \Int{1:n}$.
The maximum is attained at $t \in \{ 1, n{-}1 \}$,
\ifISIT
\else
which (with $\sin \alpha = 2 \sin (\alpha/2) \cos (\alpha/2)$)
yields~(\ref{eq:r=2,m=2}).
By Corollary~\ref{cor:m-height-MDS} and Eq.~(\ref{eq:r=2}),
one of the $2$-extremal codewords is all $1$'s
except for its first two entries, which are given by:
\[
c_0 = \Height_2(\code(n))
\quad \textrm{and} \quad
c_1 = -\Height_2(\code(n)) .
\]

\fi
For $m = 1$, we substitute $\Subset = \{ 0 \}$
and $\JJ = \{ 0, t \}$ in~(\ref{eq:m-height-dual2}), which
means that we need now
to minimize~(\ref{eq:r=2,t}) over $t \in \Int{1:n}$.
\ifISIT
\else
When $n$ is even, the minimum is attained at $t = n/2$,
with a $1$-extremal codeword which is all $1$'s except
for the following two entries:
\[
c_0 = 0
\quad \textrm{and} \quad
c_{n/2} = -\Height_1(\code(n)) .
\]
When $n$ is odd, the minimum is attained at
$t = (n \pm 1)/2$
with a $1$-extremal codeword\footnote{%
Strictly speaking, these particular $1$-extremal codewords are
negacyclic shifts of the codewords that we get from
Corollary~\ref{cor:m-height-MDS} and Eq.~(\ref{eq:r=2}).}
which is all $1$'s except for the middle entry, which is:
\[
c_{(n-1)/2} = -\Height_1(\code(n)) .
\]
These results, in turn, establish~(\ref{eq:r=2,m=1}).\qed
\fi
\end{example}

\ifISIT
\else
\begin{remark}
\label{rem:r=2}
The symmetry of the code $\code(n)$ in Example~\ref{ex:r=2}
leads to conjecture that it may have the largest $2$-height
among all linear $[n,n{-}2]$ codes over $\Realfield$.
With respect to the $1$-height, however, a better code is obtained by
applying the simple construction in \cite[\S III.B]{Roth1}
with $r = 2$: the $2 \times n$ parity-check matrix then consists of
$\lfloor n/2 \rfloor$ copies of the column vector $(1 \; 0)^\top$
and $\lceil n/2 \rceil$ copies of $(0 \; 1)^\top$.
The $1$-height in this case is $\lceil n/2 \rceil - 1$,
compared to (at least) $\cot (\pi/(2n)) - 1 \approx (2/\pi) n  - 1$
for the code $\code(n)$.\qed
\end{remark}

\subsection{Height at the redundancy of the code}
\label{sec:r-height}

In this section, we look at the $m$-height of
a linear $[n,k{=}n{-}r,d]$ code $\code$ when $m = r$;
recall by~(\ref{eq:height-distance})
that the $r$-height is finite
if and only if $r \le d - 1$, which happens
(with equality) if and only if $\code$ is MDS.

The next corollary follows from Lemma~\ref{lem:structural}.

\begin{corollary}
\label{cor:structural}
Let $\code$ be a linear $[n,k{=}n{-}r]$ MDS code over $\Realfield$.
Then for any $r$-extremal codeword $\bldc$ in $\code$,
\[
\PP_r(\bldc) \supseteq \Int{r:n}
\]
and, in particular, $\weight(\bldc) = n$.
\end{corollary}

In the Appendix,
we present an alternate proof of Corollary~\ref{cor:structural}
that uses Theorem~\ref{thm:LP} directly without
relying on the theory of duality of linear programming.

\begin{remark}
\label{rem:structural}
It follows from Corollary~\ref{cor:structural}
that any shortening of a linear $[n,k{=}n{-}r]$ MDS code
strictly decreases the $r$-height of the code.\qed
\end{remark}

The next corollary follows from Theorem~\ref{thm:m-height-dual}.

\begin{corollary}
\label{cor:r-height}
Let $\code$ be a linear $[n,k{=}n{-}r]$ MDS code over $\Realfield$.
Then the following holds.
\begin{list}{}{\settowidth{\labelwidth}{\textrm{(iii)}}}
\item[(i)]
For any $k \times n$ generator matrix
$G = (\bldg_0 \; \bldg_1 \; \ldots \; \bldg_{n-1})$ of $\code$
over $\Realfield$,
\begin{equation}
\label{eq:r-height1}
\Height_r(\code) =
\max_\Subset
\max_{i \in \Subset}
\,
\bigl\| (G)_{\overline{\Subset}}^{-1} \bldg_i \bigr\|_1 ,
\end{equation}
where $\Subset$ ranges over all $r$-subsets of $\Int{n}$.
\item[(ii)]
For any $r \times n$ parity-check matrix $H$ of $\code$
over $\Realfield$,
\[
\Height_r(\code) =
\max_\Subset
\max_{i \in \Subset}
\,
\left\|
\bigl[ (H)_\Subset^{-1} \bigr]_i (H)_{\overline{\Subset}}
\right\|_1 ,
\]
where $\Subset$ is as in part~(i).
\item[(iii)]
Denoting by $\code_{\min}^\perp$ the set of all minimum-weight codewords
in $\code^\perp$ that have at least one entry equaling $1$,
\[
\Height_r(\code) =
-1 + \max_{\bldc \in \code_{\min}^\perp} \| \bldc \|_1 .
\]
\end{list}
\end{corollary}

\begin{proof}
(i)+(ii)~When $m = r$ (and $\code$ is MDS),
the inner minima in~(\ref{eq:m-height-dual1})
and~(\ref{eq:m-height-dual2}) are over sets of size $1$.

(iii)~We observe
that as $\Subset$ ranges over all the $r$-subsets of $\Int{n}$,
the rows of $(H)_\Subset^{-1} H$ range over all
the minimum-weight codewords in $\code_{\min}^\perp$
(a codeword $\bldc$ may appear as a row for
multiple $r$-subsets---one subset for each entry of $\bldc$
that equals $1$).
\end{proof}

\begin{remark}
\label{rem:r-height}
With the code $\code$ as in Corollary~\ref{cor:r-height},
every $(k{+}1)$-subset $\XX$ of $\Int{n}$
defines a one-dimensional subspace of $\code^\perp$
whose nonzero elements are minimum-weight codewords of $\code^\perp$
with support $\XX$ (and, as $\XX$ ranges over
all $(k{+}1)$-subsets of $\Int{n}$, we get this way
all the minimum-weight codewords of $\code^\perp$).
Hence,
\[
\left| \code_{\min}^\perp \right| \le (n-r+1) \cdot \binom{n}{n-r+1}
= r \cdot \binom{n}{r} .
\]
The right-hand side expression, in turn, is
what we get in~(\ref{eq:complexity-dual}) when $m = r$.\qed
\end{remark}

In the Appendix, we present
an alternate proof of Corollary~\ref{cor:r-height}(i)
that does not rely on the theory of duality of linear programming.
\fi

\section{Spherical parity-check codes and their duals}
\label{sec:spherical}

In this section, we consider codes with parity-check matrices $H$
whose columns are all unit vectors,
namely $\diag (H^\top H) = \bldone$.
We will say that a matrix over $\Realfield$ is \emph{spherical}
if it satisfies this property,
and a linear code over $\Realfield$ that has
such a parity-check matrix will be called
a \emph{spherical parity-check} (in short, SPC) code.

The code in Example~\ref{ex:r=2} is an SPC code,
and additional SPC codes have been proposed and
analyzed in~\cite{Roth1} and~\cite{WeiRoth}.
One of the constructions therein is a family
of linear $[n,n{-}r]$ codes over $\Realfield$
\ifISIT
\else
which are defined
by means of a linear $[r,\kappa,\partial]$ code $B$ over $\Field_2$
that contains the codeword $\bldone$ and, in addition,
its dual code, $B^\perp$, has minimum distance${} > 2$.
The induced linear $[n,n{-}r]$ code, $\code_B$, over $\Realfield$
has length $n = 2^{\kappa-1}$
and an $r \times n$ parity-check matrix $H_B$ whose columns are
\[
(1/\sqrt{r}) \cdot \bigl( (-1)^{x_j} \bigr)_{j \in \Int{r}} ,
\]
where $\bldx = (x_j)_{j \in \Int{r}}$ ranges over all
the codewords in $B$ that start with a $0$.
The set of columns of $H_B$,
in turn, forms a spherical code in $\Realfield^r$
whose minimum angle between its elements (i.e., its packing property)
increases with $\partial$ and, consequently,
for every $m \le \lceil r/(r{-}2\partial) \rceil$,
the $m$-height of $\code_B$ decreases with $\partial$:
see Theorem~14 in~\cite{Roth1} and Theorems~11 and~13 in~\cite{WeiRoth}.
This allows us to speculate that, as a rule of thumb,
good packing spherical codes in $\Realfield^r$
yield SPC codes with good (i.e., small) $m$-heights.

The requirements from the code $B$ imply that the $r$ rows of $H_B$
are orthogonal and have the same $L_2$ norm (which then has to be
$\sqrt{n/r}$, since all the columns in $H_B$ are unit vectors);
equivalently, $H_B \, H_B^\top = (n/r) \cdot I_r$.
\fi
We say that an $r \times n$ matrix $H$ is \emph{ortho-spherical} if
\[
\diag(H^\top H) = \bldone
\quad \textrm{and} \quad
H \, H^\top = (n/r) \cdot I_r ,
\]
and a linear $[n,n{-}r]$ code over $\Realfield$
that has such a parity-check matrix will be called
an \emph{ortho-spherical parity-check} (OSPC) code.
The code in Example~\ref{ex:r=2} is, too, an OSPC code.
In frame theory and signal processing, (the set of columns of)
an $r \times n$ ortho-spherical matrix $H$ is equivalent
to a \emph{unit-norm tight frame},
and the scaled matrix $\sqrt{r/n} \cdot H$ is equivalent
to an \emph{equal-norm Parseval frame};
see~\cite{CFMPS},\cite{CKP},\cite{Daubechies},\cite{Waldron}.

The next (known) lemma states that the OSPC property is preserved under
duality.
\ifISIT
\else
    For completeness, we include a proof in the Appendix.
\fi

\begin{lemma}[{\cite[Proposition~3.1(c)]{CFMPS}}]
\label{lem:orthogonal}
A linear code is an OSPC code, if and only if its dual code is.
\end{lemma}

\ifISIT
\else
Letting $G$ be a $k \times n$ ortho-spherical generator matrix of
an $[n,k]$ OSPC code $\code$, we have
$G \, G^\top  = (n/k) \cdot I_k$; hence,
for every $\bldu \in \Realfield^k$,
\[
\| \bldu G \|_2^2 = \bldu G \, G^\top \bldu^\top
= (n/k) \cdot \| \bldu \|_2^2 ,
\]
namely, a (non-systematic) encoder
$\bldu \mapsto \sqrt{k/n} \cdot \bldu G$ is norm reserving.
\fi

The following two examples present additional OSPC codes.

\begin{example}
\label{ex:icosahedron}
Let $\code_\ico$ be the $[6,3]$ MDS code with the parity-check matrix:
{%
\[
H =
\frac{1}{\sqrt{\varphi+2}}
\left(
\begin{array}{rrrrrr}
0       & \phantom{-}\varphi & 1 & 0        & 1         & -\varphi \\
1       & 0 & \phantom{-}\varphi & 1        & -\varphi  & 0        \\
\varphi & 1 & 0                  & -\varphi & 0         & 1
\end{array}
\right) ,
\]%
}%
where $\varphi$ stands for the golden ratio $(1 + \sqrt{5})/2$.
The columns of $H$ range over six out of the $12$ vertices---one from
each antipodal pair---of a regular icosahedron in $\Realfield^3$
that is inscribed in the unit sphere $\Sphere^2$.
Moreover, it is easy to check that $H \, H^\top = \sqrt{2} \cdot I_3$,
namely, $\code_\ico$ is a $[6,3]$ OSPC code,
and so is its dual $\code_\ico^\perp$.
In fact, $\code_\ico$ is ``skew self-dual'':
a generator matrix can be obtained by negating
the first three columns of $H$ and then swapping them with
the last three columns.
Table~\ref{tab:icosahedron} presents the $m$-heights of $\code_\ico$
(and of $\code_\ico^\perp$), as computed by
Theorems~\ref{thm:m-height}
\ifISIT
\else
or~\ref{thm:m-height-dual}
(see also Example~\ref{ex:icosahedron-dual} below).\qed
\fi
\end{example}
\begin{table}[hbt]
\caption{$m$-heights of the icosahedral code $\code_\ico$
in Example~\protect\ref{ex:icosahedron}.}
\label{tab:icosahedron}
\ifISIT
\else
\[
\arraycolsep3ex
\begin{array}{ccc}
\hline\hline
m     & \Height_m(\code_\ico) & \textrm{approx.} \\
\hline
1, 2  & \sqrt{5}              & 2.2361 \\
3     & 2 +  \sqrt{5}         & 4.2361 \\
\hline\hline
\end{array}
\]%
\fi
\end{table}

\begin{example}
\label{ex:dodecahedron}
Let $\code_\dod$ be the $[10,7]$ MDS code with the parity-check matrix:
{%
\[
H =
\frac{1}{\sqrt{3}}
\left(
\arraycolsep0.6ex
\begin{array}{rrrrrrrrrr}
1  & 1  & 1  & 1  &
0       & \phantom{-}\psi   & \varphi & 0       & -\psi   & \varphi \\
1  & 1  & -1 & -1 &
\phantom{-}\psi   & \varphi & 0       & -\psi   & \varphi & 0       \\
1  & -1 & 1  & -1 &
\varphi & 0       & \phantom{-}\psi   & \varphi & 0       & -\psi
\end{array}
\right) ,
\]%
}%
where $\varphi = (1 + \sqrt{5})/2$ and
$\psi = -1/\varphi = (1 - \sqrt{5})/2$.
The columns of $H$ range over ten out of the $20$ vertices
(one from each antipodal pair) of
a regular dodecahedron in $\Realfield^3$
that is inscribed in $\Sphere^2$; also,
$H \, H^\top = \sqrt{10/3} \cdot I_3$
and, so, $\code_\dod$ is an OSPC code.
Table~\ref{tab:dodecahedron} presents the $m$-heights of $\code_\dod$
and Table~\ref{tab:dodecahedron-dual} presents the $m$-heights of
its $[10,3]$ dual
\ifISIT
\else
code (see also Example~\ref{ex:dodecahedron-dual} below).\qed
\fi
\end{example}
\begin{table}[hbt]
\caption{$m$-heights of the dodecahedral code $\code_\dod$
in Example~\protect\ref{ex:dodecahedron}.}
\label{tab:dodecahedron}
\ifISIT
\else
\[
\arraycolsep3ex
\begin{array}{ccc}
\hline\hline
m     & \Height_m(\code_\dod) & \textrm{approx.} \\
\hline
1     & 2 +  \sqrt{5}         & 4.2361  \\
2     & 4 +  \sqrt{5}         & 6.2361  \\
3     & 9 + 4 \sqrt{5}        & 17.9443 \\
\hline\hline
\end{array}
\]%
\fi
\end{table}
\begin{table}[hbt]
\caption{$m$-heights of the dual dodecahedral code $\code_\dod^\perp$
in Example~\protect\ref{ex:dodecahedron}.}
\label{tab:dodecahedron-dual}
\ifISIT
\else
\[
\arraycolsep3ex
\begin{array}{ccc}
\hline\hline
m     & \Height_m(\code_\dod^\perp)\rule[-1.2ex]{0ex}{3.4ex}
                                    & \textrm{approx.} \\
\hline
1     & 3/\sqrt{5}                  & 1.3416 \\
2     & (1+ \sqrt{5})/2             & 1.6180 \\
3     & 4 -  \sqrt{5}               & 1.7639 \\
4     & 3                           & 3.0000 \\
5, 6  & 2 +  \sqrt{5}               & 4.2361 \\
7     & 5 + 2\sqrt{5}               & 9.4721 \\
\hline\hline
\end{array}
\]%
\fi
\end{table}

\ifISIT
\else
\begin{remark}
\label{rem:platonic}
In addition to the dodecahedron and the icosahedron,
there are three more platonic solids, yet they do not
yield interesting constructions: both the tetrahedron and the cube yield
the $[4,1]$ repetition code, and the octahedron yields the trivial
$[3,0]$ code.\qed
\end{remark}

\begin{remark}
\label{rem:ruthishauer}
A construction due to Ruthishauser~\cite{Ruthishauser}
yields a $20$-vertex polyhedron with a minimum angle
(between vertices) that is larger
than that in the regular dodecahedron. This polyhedron has
several symmetries, including closure under antipodality.
Let $\code_\rut$ be the linear $[10,7]$ code
with a $3 \times 10$ parity-check matrix whose columns are obtained
by selecting one vertex from each antipodal pair in
Ruthishauser's construction.
The code $\code_\rut$ turns out to be an OSPC code, yet it is not MDS
(neither is its $[10,3]$ dual $\code_\rut^\perp$) and, therefore,
$\Height_3(\code_\rut) = \Height_7(\code_\rut^\perp) = \infty$.
Still,
$\Height_1(\code_\rut) = 1 + 2 \sqrt{2} \approx 3.8284$ and
$\Height_2(\code_\rut) = 3 + 2 \sqrt{2} \approx 5.8284$,
which are smaller than the respective heights of $\code_\dod$,
and
$\Height_4(\code_\rut^\perp) = \sqrt{2/3} + \sqrt{3} \approx 2.5485$ and
$\Height_5(\code_\rut^\perp) = \sqrt{3/2} + \sqrt{3} \approx 2.9567$,
which are smaller than the respective heights of $\code_\dod^\perp$.
An improvement to Ruthishauser's construction due
to van der Waerden~\cite{Waerden} breaks its symmetries,
including the closure under antipodality.\qed
\end{remark}
\fi

We now turn to consider the duals of SPC codes
(including all OSPC codes).
Let $\code$ be such a linear $[n,k,d]$ code over $\Realfield$,
namely, $\code$ has a spherical $k \times n$ generator matrix
$G = (\bldg_j)_{j \in \Int{n}}$. We associate with $G$ the multiset
\begin{equation}
\label{eq:GG}
\GG = \left\{ \bldg_j \right\}_{j \in \Int{n}} \cup
\left\{ -\bldg_j \right\}_{j \in \Int{n}} ,
\end{equation}
which contains $2n$ vectors (or points) in $\Sphere^{k-1}$
(which are not necessarily distinct).
Also, let $m \in \Int{1:d}$.
\ifISIT
\else

Consider any codeword
$\bldc = (c_j)_{j \in \Int{n}} = \bldu G$
where (without loss of generality) $\bldu \in \Sphere^{k-1}$.
Then,
\[
c_j = \bldu \cdot \bldg_j , \quad j \in \Int{n} .
\]
Letting $\pi(\cdot)$ be a sorting permutation of $\bldc$, we have
\[
\Height_m(\bldc)
=
\left|
\frac{\bldu \cdot \bldg_{\pi(0)}}{\bldu \cdot \bldg_{\pi(m)}}
\right| .
\]
Thus, to get a small $m$-height of $\code$,
we would like $|\bldu \cdot \bldg_{\pi(m)}|$
to be sufficiently close to $|\bldu \cdot \bldg_{\pi(0)}|$
for every vector $\bldu \in \Sphere^{k-1}$.
Yet $|\bldu \cdot \bldg_j|$ is the cosine of
the angle between $\bldu$ and either
$\bldg_j$ or $-\bldg_j$.
This means that for each vector $\bldu \in \Sphere^{k-1}$,
we would like to see small disparity among
the (Euclidean) distances between $\bldu$ and the $m+1$ points in $\GG$
that are the closest to $\bldu$.
One possible strategy to achieve this goal is requiring
that $\GG$ forms a good
$(m{+}1)$-fold $\theta$-covering spherical code\footnote{%
Earlier, we mentioned the relation between
the height profile of an SPC code and
the \emph{packing} properties of the spherical code
that is formed by the columns of its spherical parity-check matrix,
which is also the generator matrix of its dual code.
Here, when considering the dual code, we refer instead
to the \emph{covering} properties of these columns.}
in $\Sphere^{k-1}$.
Under these conditions, every vector $\bldu \in \Sphere^{k-1}$ is at
angular distance${} \le \theta$ from
at least $m+1$ points in $\GG$ and, so,
\[
\Height_m(\code) \le \frac{1}{\cos \theta} .
\]

For a linear $[n,k,d]$ code $\code$ whose dual code
is an $[n,n{-}k,d^\perp]$ SPC code and for $m \in \Int{1:d}$,
Lemma~\ref{lem:structural} implies\footnote{%
Since $\code^\perp$ is an SPC code, none of its parity-check matrices
has an all-zero column and, therefore, the condition $d^\perp \ge 2$
in the lemma is met. For the case $d^\perp = 2$,
the upcoming discussion becomes trivial.}
\fi
the following
necessary geometric condition on the $m$-extremal codewords of $\code$.
If $\bldc = \bldu G$ is such a codeword, then there must be
a $(d^\perp{-}1)$-subset $\GG_\bldu \subseteq \GG$ such that
the dot products $\bldu \cdot \bldg$
are equal for all $\bldg \in \GG_\bldu$;
moreover, the elements in $\GG_\bldu$ are linearly
\ifISIT
\else
independent,
thereby defining a $(d^\perp{-}2)$-dimensional affine hyperplane
$\AA_\bldu$ in $\Realfield^k$.
This means that the point $\bldu$ is equidistant---in Euclidean
as well as angular terms---from all the points in $\GG_\bldu$.
Equivalently, $\bldu$ points toward
the center of the (unique) hypersphere in $\AA_\bldu$
that passes through all the points of $\GG_\bldu$.

Similarly, by Lemma~\ref{lem:m-height}, there exists
an $m$-extremal codeword $\bldc = \bldu G$ where $\bldu$
is equidistant from $k$ linearly independent points in $\GG$.
Here, $\bldu$ points toward the center of the hypersphere in
the $(k{-}1)$-dimensional affine hyperplane
that passes through these $k$ points.
(Recall that when $\code$ is MDS then $d^\perp = k+1$
and so, by Lemma~\ref{lem:structural}, this condition on $\bldu$
must hold for every $m$-extremal codeword $\bldu G$.)

The next theorem summarizes our previous discussion.
\fi

\begin{theorem}
\label{thm:SPC-structural}
Let $\code$ be a linear $[n,k,d]$ code over $\Realfield$
whose dual is an $[n,n{-}k,d^\perp]$ SPC code,
let $G$ be a $k \times n$ spherical generator matrix of $\code$,
and let $\GG$ be defined as in~(\ref{eq:GG}).
\ifISIT
\else
Then the following holds for every $m \in \Int{1:d}$.
\begin{list}{}{\settowidth{\labelwidth}{\textrm{(ii)}}}
\item[(i)]
For any $m$-extremal codeword $\bldc = \bldu G$ in $\code$,
the point $\bldu$ (in $\Realfield^k$) is equidistant from at least
$d^\perp - 1$ linearly independent points in $\GG$.
\item[(ii)]
There exists an $m$-extremal codeword $\bldc = \bldu G$ in $\code$
such that the point $\bldu$ is equidistant from at least
$k$ linearly independent points in $\GG$.
\end{list}
\fi
\end{theorem}

\begin{example}
\label{ex:k=2}
Let $\code^\perp(n)$ be the linear $[n,2,n{-}1]$ code which is
the dual of the negacyclic code $\code(n)$ in Example~\ref{ex:r=2},
let $G = (\bldg_j)_{j \in \Int{n}}$ be
its $2 \times n$ generator matrix whose columns are given
by~(\ref{eq:hj}), and let $\GG$ be defined as in~(\ref{eq:GG}).
The $2n$ points in $\GG$ form the vertices of a regular $(2n)$-gon
in $\Realfield^2$. By Theorem~\ref{thm:SPC-structural},
for any $m \in \Int{1:n{-}1}$,
the point $\bldu \in \Realfield^2$ that corresponds to
an $m$-extremal codeword $\bldu G$
must be equidistant from two vertices of this polygon.
Due to the $(2n)$-fold rotational symmetry of $\GG$,
the vector $\bldu$ either points toward one of the vertices of $\GG$
\ifISIT
\else
(say, $\bldg_0$)
\fi
or toward the midpoint of an edge of
\ifISIT
\else
the polygon (say, the edge with endpoints $\bldg_0$ and $\bldg_1$).

Letting $\beta = \alpha/2 = \pi/(2n)$, in the first case we get
\[
\bldc_1
= \bigl( 1 \;\; 0 \bigr) \, G
= \bigl( \cos (2 j \beta) \bigr)_{j \in \Int{n}}
\]
and, respectively,
\begin{equation}
\label{eq:case1}
\Height_m(\bldc_1) = \frac{1}{\cos (2 \lceil m/2 \rceil \beta)} ,
\quad
m \in \Int{n{-}1} .
\end{equation}
In the second case,
\[
\bldc_2
= \bigl( \cos \beta \;\; \sin \beta \bigr) \, G
= \bigl( \cos ((2j{-}1) \beta) \bigr)_{j \in \Int{n}}
\]
and, respectively,
\begin{equation}
\label{eq:case2}
\Height_m(\bldc_2)
= \frac{\cos \beta}{\cos ((2 \lfloor m/2 \rfloor{+}1) \beta)} ,
\quad
m \in \Int{n{-}1} .
\end{equation}
Hence, when $m$ is odd,
\[
\Height_m(\bldc_1)
\stackrel{\textrm{(\ref{eq:case1})}}{=}
\frac{1}{\cos ((m{+}1) \beta)}
>
\frac{\cos \beta}{\cos (m \beta)}
\stackrel{\textrm{(\ref{eq:case2})}}{=}
\Height_m(\bldc_2)
\]
(since $1 > \cos \beta$ and $\cos ((m{+}1) \beta) < \cos (m \beta)$),
and when $m$ is even,
\[
\Height_m(\bldc_1)
\stackrel{\textrm{(\ref{eq:case1})}}{=}
\frac{1}{\cos (m \beta)}
<
\frac{\cos \beta}{\cos ((m{+}1) \beta)}
\stackrel{\textrm{(\ref{eq:case2})}}{=}
\Height_m(\bldc_2)
\]
(since
$\cos ((m{+}1) \beta)
= (\cos \beta) \cos (m \beta) - (\sin \beta) \sin (m \beta)
< (\cos \beta) \cos (m \beta)$).

We conclude that $\bldc_1$ (respectively, $\bldc_2$),
as well as each of its negacyclic shifts, is $m$-extremal
for any odd (respectively, even) $m \in \Int{n{-}1}$, namely,
\fi
\begin{equation}
\ifISIT
\else
\label{eq:k=2}
\fi
\Height_m(\code^\perp(n))
=
\left\{
\ifISIT
\else
\renewcommand{\arraystretch}{2.0}
\fi
\begin{array}{l@{\quad\;}l}
\ifISIT
\else
\displaystyle
\frac{1}{\cos ((m{+}1) \beta)} ,
\fi
& \textrm{when $m$ is odd} \\
\ifISIT
\else
\displaystyle
\frac{\cos \beta}{\cos ((m{+}1) \beta)} ,
\fi
& \textrm{when $m$ is even} .
\end{array}
\right.
\end{equation}%
\qed
\end{example}

\ifISIT
    Theorem~\ref{thm:SPC-structural} can be used to characterize
    explicitly the extremal codewords, $\bldc = \bldu G$,
    of the codes $\code_\ico$ and $\code_\dod^\perp$ as well
    (we omit the details).
    For these codes (including the code $\code^\perp(n)$),
    the characterization of the extremal codewords
\else
Interestingly, the height profile of the code $\code^\perp(n)$ in
the example is related to its primal counterpart
in Example~\ref{ex:r=2} by:
\[
\Height_{n-2}(\code^\perp(n)) = \Height_1(\code(n)) + 1 .
\]

\begin{remark}
\label{rem:k=2a}
For many values of $m \in \Int{n{-}1}$,
the code $\code^\perp(n)$ in Example~\ref{ex:k=2}
does not have the smallest $m$-height
among all linear $[n,2]$ codes.
For example, by (possibly puncturing)
the construction in~\cite[\S III.A]{Roth1},
for any $k \in \Int{2:n}$ we get linear $[n,k]$ codes
whose $m$-height is $1$ for any $m \le n/k - 1$;
in particular,
their $(\lfloor n/2 \rfloor {-}1)$-height is $1$ when $k = 2$.

Better codes exist also for larger $m$.
For example, suppose that $n$ is divisible by $4$ and let
$\tilde{\code}$ be the linear $[n,2]$ code
that is obtained from the construction $\code^\perp(n/2)$
by duplicating every coordinate in the codewords of the latter.
For $m = n-3$ (which is odd) we have
\begin{equation}
\label{eq:k=2,m=n-3}
\Height_m(\code^\perp(n))
\stackrel{\mathrm{(\ref{eq:k=2})}}{=}
\frac{1}{\cos ((n{-}2) \pi/(2n))}
= \frac{1}{\sin (\pi/n)} .
\end{equation}
On the other hand,
$\lfloor m/2 \rfloor = n/2 - 2$ is even and, so,
\begin{eqnarray*}
\Height_m(\tilde{\code})
& = &
\Height_{\lfloor m/2 \rfloor}(\code^\perp(n/2)) \\
& \stackrel{\mathrm{(\ref{eq:k=2})}}{=} &
\frac{\cos (\pi/n)}{\cos ((n/2-1) \pi/n)} \\
& = &
\cot (\pi/n) ,
\end{eqnarray*}
which is smaller than~(\ref{eq:k=2,m=n-3}).\qed
\end{remark}

\begin{example}
\label{ex:icosahedron-dual}
Let $\code_\ico^\perp$ be the $[6,3]$ dual of the MDS code $\code_\ico$
in Example~\ref{ex:icosahedron} ($\code_\ico^\perp$ is identical
to $\code_\ico$ up to permutation and negation of coordinates).
Let $G = (\bldg_j)_{j \in \Int{6}}$ be
the $3 \times 6$ generator matrix of $\code_\ico^\perp$ which is given
by the ortho-spherical matrix $H$ in that example
and let $\GG$ be defined as in~(\ref{eq:GG}), namely, $\GG$
consists of the $12$ vertices of a regular icosahedron.

By Theorem~\ref{thm:SPC-structural}, any point
$\bldu \in \Realfield^3$ that corresponds to an $m$-extremal codeword
$\bldu G$ of $\code_\ico^\perp$ must be equidistant from
three linearly independent vertices in $\GG$; equivalently,
the vector $\bldu$ must point toward the circumcenter
of a triangle formed by three vertices of $\GG$,
no two of which are antipodes.
There are $\binom{6}{3} \cdot 2^3 = 160$ ways to select such
triples, yet we will get the same codeword multiple times.
For example, consider some vertex $\pm \bldg_\ell \in \GG$;
its five neighbors in the icosahedron form a regular pentagon.
There are $10$ ways to form a triangle when selecting any three of
these five neighbors, yet these triangles will all have the same
circumcenter (and the same circumscribed circle)
and will therefore result in the same codeword (up to scaling).
In fact, $\bldu$ in this case will point toward the vertex
$\pm \bldg_\ell$ (i.e., it will be a scalar multiple of $\bldg_\ell$).
Moreover, we will get the same set of codewords if we consider instead
the five neighbors of the antipodal vertex $\mp \bldg_\ell$.
After eliminating such replications, we are left with $16$ codewords,
which are all extremal,
and Table~\ref{tab:icosahedron-dual-codewords} shows
their breakdown according to their $m$-extremality.
When describing the properties of
each $m$-extremal codeword $\bldc = \bldu G$ (last column in the table),
we assume that $\bldc$ (or $\bldu$) is scaled
so that $|c_{\pi(m)}| = 1$.
The respective $m$-heights are found in
Table~\ref{tab:icosahedron}.\qed
\end{example}
\begin{table*}[hbt]
\caption{Properties of the $m$-extremal codewords of
the (dual) icosahedral code $\code_\ico^\perp$
in Example~\protect\ref{ex:icosahedron-dual}.}
\label{tab:icosahedron-dual-codewords}
\centering
\tabcolsep1ex
\renewcommand{\arraystretch}{1.5}
\ifIEEE
    \newcommand{\Lift}{\vspace{-2.5ex}}
    \begin{tabular}{ccp{51ex}@{\quad}p{69ex}}
\else
    \vspace{2ex}
    \scriptsize
    \newcommand{\Lift}{\vspace{-6ex}}
    \begin{tabular}{ccp{42ex}@{\quad}p{65ex}}
\fi
\hline\hline
$m$   & \#\ codewords
& \multicolumn{1}{c}{The vector $\bldu$ points toward ---}
& \multicolumn{1}{c}{Properties of $\bldc = \bldu G$} \\
\hline
$1, 2$  & $6$             &
a vertex $\pm \bldg_\ell$ among the $12$ vertices of the icosahedron
(i.e., $\bldu$ is a scalar multiple of $\bldg_\ell$).
&
\Lift
\begin{itemize}
\itemsep0ex
\item
The entry $c_\ell$ has magnitude $\sqrt{5}$.
\item
Five entries have magnitude $1$
(corresponding to the five neighbors of $\pm \bldg_\ell$
in the icosahedron).
\end{itemize}
\\ \hline
$3$  & $10$    &
the circumcenter of a face $\p{F}$ among the $20$ faces of
the icosahedron. It also points toward
the circumcenter of the triangle $\p{T}$ formed
by the three vertices outside $\p{F}$
that are on the three faces which are adjacent to $\p{F}$.
&%
\Lift
\begin{itemize}
\itemsep0ex
\item
Three entries have magnitude $2 + \sqrt{5}$
(corresponding to the three vertices of the face $\p{F}$).
\item
Three entries have magnitude $1$
(corresponding to the three vertices of the triangle $\p{T}$).
\end{itemize}
\\
\hline\hline
\end{tabular}
\end{table*}

\begin{example}
\label{ex:dodecahedron-dual}
Let $\code_\dod^\perp$ be the $[10,3]$ dual code of
the MDS code $\code_\dod$ in Example~\ref{ex:dodecahedron}
and let $G = (\bldg_j)_{j \in \Int{10}}$ be
the $3 \times 10$ generator matrix that is given by
the ortho-spherical matrix $H$ in that example.
Also, let $\GG$ be defined as in~(\ref{eq:GG}), namely, $\GG$
consists of the $20$ vertices of a regular dodecahderon.
Here, there are $\binom{10}{3} \cdot 2^3 = 960$
ways to select triples from $\GG$ that form
nondegenerate triangles, and after pruning replications, we are left
with $121$ codewords, out of which $30$ codewords are not extremal.
The breakdown of the remaining $91$ codewords, according
to their $m$-extremality, is shown in\footnote{%
The description in the last column in the table assumes that
$|c_{\pi(m)}| = 1$, except when $m = 1$, in which case
$|c_{\pi(1)}| = \sqrt{5}$.}
Table~\ref{tab:dodecahedron-dual-codewords},
and the respective $m$-heights are found in
Table~\ref{tab:dodecahedron-dual}.\qed
\end{example}
\begin{table*}[hbt]
\caption{Properties of the $m$-extremal codewords of
the dual dodecahedral code $\code_\dod^\perp$
in Example~\protect\ref{ex:dodecahedron-dual}.}
\label{tab:dodecahedron-dual-codewords}
\centering
\tabcolsep1ex
\renewcommand{\arraystretch}{1.5}
\ifIEEE
    \newcommand{\Lift}{\vspace{-2.5ex}}
    \begin{tabular}{ccp{51ex}@{\quad}p{69ex}}
\else
    \vspace{2ex}
    \scriptsize
    \newcommand{\Lift}{\vspace{-6ex}}
    \begin{tabular}{ccp{42ex}@{\quad}p{65ex}}
\fi
\hline\hline
$m$   & \#\ codewords
& \multicolumn{1}{c}{The vector $\bldu$ points toward ---}
& \multicolumn{1}{c}{Properties of $\bldc = \bldu G$} \\
\hline
$1, 4$  & $10$             &
a vertex $\pm \bldg_\ell$ among the $20$ vertices of the dodecahedron.
&
\vspace{0.5ex}
\Lift
\begin{itemize}
\itemsep0ex
\item
The entry $c_\ell$ has magnitude $3$.
\item
Three entries have magnitude $\sqrt{5}$
(corresponding to the three neighbors of $\pm \bldg_\ell$
in the dodecahedron).
\item
Six entries have magnitude $1$
(corresponding to the neighbors of these three neighbors).
\end{itemize}
\\ \hline
$5, 6$  & $6$    &
the circumcenter of a face $\p{F}$ among the $12$ faces of
the dodecahedron.
&%
\Lift
\begin{itemize}
\itemsep0ex
\item
Five entries have magnitude $2 + \sqrt{5}$
(corresponding to the five vertices of the face $\p{F}$).
\item
Five entries have magnitude $1$
(corresponding to the five neighbors of the vertices of $\p{F}$
that are not on $\p{F}$).
\end{itemize}
\\ \hline
$2$     & $15$             &
the midpoint of an edge $e$ among the $30$ edges of the dodecahedron.
As such, $\bldu$ points toward the circumcenter of
the square $\p{Q}$ formed by the four neighbors
of the two endpoints of $e$, and is also perpendicular to
a plane that passes through two pairs of antipodal vertices.
&%
\Lift
\begin{itemize}
\itemsep0ex
\item
Two entries have magnitude $\varphi = (1 + \sqrt{5})/2$
(corresponding to the two endpoints of the edge $e$).
\item
Four entries have magnitude $1$
(corresponding to the vertices of the square $\p{Q}$).
\item
Two entries have magnitude $1/\varphi$.
\item
Two entries are zero (i.e., $\bldc$ is a minimum-weight codeword).
\end{itemize}
\\ \hline
$3$     & $30$    &
the circumcenter of an isosceles $\Delta \p{A}\p{B}\p{C}$
with a base $\p{B}\p{C}$ which is constructed as follows.
Select a face $\p{F}$ and let $-\p{A}$ be
a vertex of $\p{F}$ (this vertex will be the antipode of $\p{A}$).
Then $\p{B}$ and $\p{C}$ are
the two (adjacent) vertices of $\p{F}$ that are not neighbors of
$-\p{A}$.
&%
\Lift
\begin{itemize}
\itemsep0ex
\item
One entry has magnitude $4 - \sqrt{5}$
(corresponding to the unique vertex that is not a neighbor
of $\p{B}$ or $\p{C}$
on the face that is adjacent to $\p{F}$ along $\p{B}\p{C}$).
\item
Three entries have magnitude $1$.
\item
The magnitudes $-3 + 2 \sqrt{5}$,
$5 - 2 \sqrt{5}$, and $-2 + \sqrt{5}$ appear twice each.
\end{itemize}
\\ \hline
$7$  & $30$    &
the circumcenter of an isosceles $\Delta \p{A}\p{B}\p{C}$
with a base $\p{B}\p{C}$ which is constructed as follows.
Select a face $\p{F}$, let $-\p{A}$ be
a vertex of $\p{F}$, and let $\p{D}$ and $\p{E}$ be
the two neighbors of $-\p{A}$ on $\p{F}$.
Then $\p{B}$ (respectively, $\p{C}$) is the unique neighbor
of $\p{D}$ (respectively, $\p{E}$) that is not on $\p{F}$.
&%
\Lift
\begin{itemize}
\itemsep0ex
\item
Two entries have magnitude $5 + 2 \sqrt{5}$
(corresponding to the antipodes of the two remaining vertices
of $\p{F}$).
\item
One entry has magnitude $4 + \sqrt{5}$.
\item
Three entries have magnitude $1$.
\item
The magnitudes $3 + 2 \sqrt{5}$ and $2 + \sqrt{5}$ appear twice each.
\end{itemize}
\\
\hline\hline
\end{tabular}
\end{table*}

For the codes $\code^\perp(n)$,
$\code_\ico^\perp$, and $\code_\dod^\perp$
in Examples~\ref{ex:k=2}--\ref{ex:dodecahedron-dual},
the characterization of the extremal codewords, $\bldu G$,
\fi
and, in particular, of the height profile of the codes,
can also be obtained directly as follows:
(a)~exploit the symmetry of the geometric objects
to limit the search space for the vectors $\bldu$ to a small domain,
(b)~characterize the order statistics of the entries of
the respective codewords $\bldu G$, and
(c)~determine via critical point analysis the specific vectors
$\bldu$ that correspond to the $m$-extremal codewords.
See~\cite{ZYRSJ} for more details.

\ifISIT
\else
\ifIEEE
   \appendix[Added proofs]
\else
   \section*{$\,$\hfill Appendix: Added proofs\hfill$\,$}
   \appendix
\fi

\begin{proof}[Proof of Lemma~\ref{lem:m=r-1}]
Write $w = \weight(\bldb)$.
Since the function $f$ remains unchanged under
(the same) permutation on the coordinates of $\blda$ and $\bldb$,
we can permute this vectors and assume that
$\Support(\bldb) = \Int{w}$ and that
\[
\frac{a_0}{b_0}
\le
\frac{a_1}{b_1}
\le
\cdots
\le
\frac{a_{w-1}}{b_{w-1}} .
\]
Writing $\rho_j = a_j/b_j$ when $j \in \Int{w}$,
the function $f$ can be written as
\[
f(v) =
\sum_{j \in \Int{w}} |b_j| \cdot |v - \rho_j|
+ \sum_{j \in \Int{w:\ell}} |a_j| .
\]
The derivative of $f(v)$ at its differentiable points is
\[
f'(v) = \sum_{j \in \Int{w}} |b_j| \cdot \sgn \, (v - \rho_j) .
\]
Hence, a minimizing $v$ is given by $\rho_h$, where $h$ is the largest
integer in $\Int{w}$ such that
\[
\sum_{j \in \Int{h}} |b_j|
\le \frac{1}{2} \sum_{j \in \Int{\ell}} |b_j|
= \frac{\| \bldb \|_1}{2}.
\]
At that minimum,
\[
f(\rho_h) = \| \blda - \rho_h \cdot \bldb \|_1 ,
\]
and $\blda - \rho_h \cdot \bldb$ is $0$ at position $h$.
\end{proof}

\begin{proof}[Alternate proof of Corollary~\ref{cor:structural}]
Let $\bldc$ be an $r$-extremal codeword in $\code$
and let $\pi(\cdot)$ be a sorting permutation of $\bldc$;
without loss of generality we assume that $|c_{\pi(r)}| = 1$.
Denoting
\[
\Subset = \left\{ \pi(t) \,:\, t \in \Int{r} \right\} ,
\]
we then have (see~(\ref{eq:codesubset})):
\begin{equation}
\label{eq:cj}
\bldc \in \code\vert_{\overline{\Subset}} .
\end{equation}
Since $\code$ is MDS, there is
a $k \times n$ generator matrix $G$ of $\code$
such that $(G)_{\overline{\Subset}}$ is
the $k \times k$ identity matrix.
Moreover, the remaining columns of $G$,
denoted by $\bldg_j$, $j \in \Subset$, have no zero entries.
    From $\bldc = (\bldc)_{\overline{\Subset}} \, G$
it follows that
\begin{equation}
\label{eq:c0}
c_j = (\bldc)_{\overline{\Subset}} \cdot \bldg_j ,
\quad
j \in \Subset .
\end{equation}

Suppose now to the contrary that
\begin{equation}
\label{eq:cn-1}
|c_{\pi(n-1)}| < 1
\end{equation}
and consider the unique codeword
$\bldx = (x_j)_{j \in \Int{n}} \in \code$ defined by
\begin{equation}
\label{eq:xinformation}
(\bldx)_{\overline{\Subset}} = \sgn \, ( \bldg_{\pi(0)}^\top )
\end{equation}
(with $\sgn(\cdot)$ applied componentwise);
the remaining entries of $\bldx$ are given by
\begin{equation}
\label{eq:xparity}
x_j = (\bldx)_{\overline{\Subset}} \, \bldg_j ,
\quad
j \in \Subset .
\end{equation}
Letting $\theta(\cdot)$ be a sorting permutation of $\bldx$, we have:
\begin{eqnarray*}
|x_{\theta(0)}|
& \ge &
|x_{\pi(0)}| \\
& \stackrel{\textrm{(\ref{eq:xparity})}}{=} &
\bigl| (\bldx)_{\overline{\Subset}} \cdot \bldg_{\pi(0)} \bigr| \\
& \stackrel{\textrm{(\ref{eq:xinformation})}}{=} &
\bigl\| \bldg_{\pi(0)} \bigr\|_1 \\
& \stackrel{\textrm{(\ref{eq:cj})+(\ref{eq:cn-1})}}{>} &
\bigl| (\bldc)_{\overline{\Subset}} \cdot \bldg_{\pi(0)} \bigr| \\
& \stackrel{\textrm{(\ref{eq:c0})}}{=} &
|c_{\pi(0)}| .
\end{eqnarray*}
On the other hand,
\[
|x_{\theta(r)}| \le 1
\]
(since there are no more than $r$ entries in $\bldx$ of
magnitude greater than $1$). We conclude that
\[
\Height_r(\bldx) = \left| \frac{x_{\theta(0)}}{x_{\theta(r)}} \right|
> |c_{\pi(0)}| = \Height_r(\bldc) ,
\]
thereby contradicting the fact that $\bldc$ is $r$-extremal.
\end{proof}

\begin{proof}[Alternate proof of Corollary~\ref{cor:r-height}(i)]
Let $\Subset$ and $i$ be maximizer in~(\ref{eq:r-height1})
and let $\bldc = (c_j)_{j \in \Int{n}}$ be
the unique codeword in $\code$
(and also in $\code\vert_{\overline{\Subset}}$)
that is defined by
\begin{equation}
\label{eq:r-height2}
\bldc = (\bldc)_{\overline{\Subset}} (G)_{\overline{\Subset}}^{-1} G ,
\end{equation}
where
\begin{equation}
\label{eq:r-height3}
(\bldc)_{\overline{\Subset}} =
\sgn \,
\bigl( (G)_{\overline{\Subset}}^{-1} \bldg_i \bigr)^\top
\end{equation}
(since $\code$ is MDS, the entries of the matrix
$\bigl( (G)_{\overline{\Subset}}^{-1} (G)_\Subset \bigr)^\top$
are all nonzero and, therefore, the entries in
$(\bldc)_{\overline{\Subset}}$ are all $\pm 1$).
Letting $\pi(\cdot)$ be a sorting permutation of $\bldc$, we have
$|c_{\pi(r)}| \le 1$ (since there are no more than $r$ entries
in $\bldc$ of magnitude greater than $1$); so,
\[
\Height_r(\bldc)
=
\left| \frac{c_{\pi(0)}}{c_{\pi(r)}} \right| \ge |c_{\pi(0)}| \\
\ge
c_i
\,\stackrel{\textrm{(\ref{eq:r-height2})+(\ref{eq:r-height3})}}{=}\,
\bigl\| (G)_{\overline{\Subset}}^{-1} \bldg_i \bigr\|_1 .
\]
Hence, the right-hand side of~(\ref{eq:r-height1}) is
a lower bound on $\Height_r(\code)$.

Conversely, let $\bldc = (c_j)_{j \in \Int{n}}$ be $r$-extremal
in $\code$ and let $\pi(\cdot)$ be a sorting permutation of $\bldc$.
Denoting
\[
\Subset = \left\{ \pi(t) \,:\, t \in \Int{r} \right\}
\]
and assuming (without loss of generality) that $|c_{\pi(r)}| = 1$,
it follows that for $i = \pi(0)$,
\[
\Height_r(\code) = \Height_r(\bldc) = |c_i|
\stackrel{\textrm{(\ref{eq:r-height2})}}{=}
\bigl| (\bldc)_{\overline{\Subset}}
(G)_{\overline{\Subset}}^{-1} \bldg_i \bigr|
\le
\bigl\| (G)_{\overline{\Subset}}^{-1} \bldg_i \bigr\|_1 ,
\]
which, in turn, is bounded from above
by the right-hand side of~(\ref{eq:r-height1}).
\end{proof}

\begin{proof}[Proof of Lemma~\ref{lem:orthogonal}]
Let $\code$ be an $[n,n{-}r]$ OSPC code over $\Realfield$
with an ortho-spherical $r \times n$ parity-check $H$.
The $r$ rows of $\sqrt{r/n} \cdot H$ form an orthonormal set,
which can be extended by $k = n-r$ vectors into
an orthonormal basis of $\Realfield^n$. Let $G$ be
the $k \times n$ matrix whose rows are these $k$ vectors,
multiplied by $\sqrt{n/k}$.
Then $G \, G^\top = (n/k) \cdot I_k$
and the $n \times n$ matrix
\[
M =
\left(
\renewcommand{\arraystretch}{1.5}
\begin{array}{c}
\sqrt{r/n}   \cdot H \\
\hline
\sqrt{k/n} \cdot G
\end{array}
\right)
\]
is orthonormal, namely,
\begin{equation}
\label{eq:orthogonal}
I_n = M^\top M =
\frac{r}{n} \cdot H^\top H + \frac{k}{n} \cdot G^\top G .
\end{equation}
Looking at the main diagonals of the matrices in~(\ref{eq:orthogonal}),
we obtain
\begin{eqnarray*}
\bldone \; = \; \diag(I_n)
& =  &
\frac{r}{n} \cdot \diag(H^\top H)
+ \frac{k}{n} \cdot \diag(G^\top G) \\
& = &
\frac{r}{n} \cdot \bldone
+ \frac{k}{n} \cdot \diag(G^\top G) ,
\end{eqnarray*}
namely,
\[
\diag(G^\top G)
= \frac{1}{k} \left( n \cdot \bldone - r \cdot \bldone \right)
= \bldone .
\]
\end{proof}
\fi

\end{document}